\documentclass[11pt]{article}

\usepackage[table, dvipsnames]{xcolor}
\usepackage{graphicx, xspace, glossaries, multirow, natbib}
\usepackage[margin=1in]{geometry}
\usepackage{amsfonts, amsmath, amsthm, amssymb, thmtools, thm-restate, mathrsfs}
\usepackage{algorithm}
\usepackage[noend]{algpseudocode}

\usepackage{hyperref}
\definecolor{mydarkblue}{rgb}{0,0.08,0.45}
\definecolor{mydarkred}{rgb}{0.75, 0.27, 0.34}
\hypersetup{
    colorlinks=true,
    citecolor=mydarkblue,
    linkcolor=mydarkred,
}

\usepackage{cleveref} 

\crefname{section}{Section}{Sections}
\crefname{subsection}{Section}{Sections}
\crefname{subsubsection}{Section}{Sections}
\crefname{appendix}{Appendix}{Appendices}

\crefname{table}{Table}{Tables}
\crefname{theorem}{Theorem}{Theorems}
\crefname{lemma}{Lemma}{Lemmas}
\crefname{conjecture}{Conjecture}{Conjectures}
\crefname{proposition}{Proposition}{Propositions}
\crefname{claim}{Claim}{Claim}
\crefname{algorithm}{Algorithm}{Algorithms}

\newtheorem{lemma}{Lemma}
\newtheorem{claim}{Claim}

\crefname{remark}{Remark}{Remarks}
\crefname{equation}{Eqn.}{Eqns.}
\crefname{equation}{Eqn.}{Eqns.}

\theoremstyle{definition}
\newtheorem{definition}{Definition}
\crefname{definition}{Definition}{Definitions}

\crefname{remark}{Remark}{Remarks}

\newacronym{lp}{LP}{Linear Program}

\newcommand{\Z}{\mathbb{Z}}
\newcommand{\R}{\mathbb{R}}
\newcommand{\E}{\mathbf{E}}

\newcommand{\cost}{\mathtt{cost}}
\newcommand{\ALG}{\mathtt{alg}}
\newcommand{\OPT}{\mathtt{opt}}

\newcommand{\kCSFE}{$k$-$\mathtt{SBFE}$\xspace}
\newcommand{\gp}{\textup{\texttt{GP}}\xspace}
\newcommand{\hgp}{\textup{\texttt{HGP}}\xspace}

\newcommand{\PSBFE}{$\mathtt{Part}$-$\mathtt{SBFC}$\xspace}
\newcommand{\SBFE}{$\mathtt{SBFE}$\xspace}
\newcommand{\SBFC}{$\mathtt{SBFC}$\xspace}

\newcommand{\matroidSBFC}{\allowdisplaybreaks\textup{\texttt{matroid}}-\textup{\texttt{SBFC}}\xspace}

\newcommand{\kmssc}{$k$-$\mathtt{MSSC}$\xspace}
\newcommand{\MSSC}{$\mathtt{MSSC}$\xspace}

\newcommand{\ugreedy}{\textup{\texttt{cond.greedy}}}
\newcommand{\convexbody}{\mathcal{K}}

\newcommand{\ignore}[1]{}

\newcommand{\mssc}{\textup{\texttt{MSSC}}\xspace}

\title{Stochastic Function {Certification} with Correlations}

\author{
    Rohan Ghuge \\ UT Austin \and
    Jai Moondra\footnote{Corresponding author, email: \texttt{jaimoondra@cmu.edu}.} \footnote{Part of this work was done while the author was at Georgia Tech.} \\ CMU \and
    Mohit Singh\footnote{The author was supported by NSF grants CCF-2504994 and CCF-2106444.} \\ Georgia Tech}

\date{}

\begin{document}

\maketitle
\thispagestyle{empty}

\begin{abstract}
    We study the \emph{Stochastic Boolean Function Certification} (\SBFC) problem, where we are given $n$ Bernoulli random variables $\{X_e: e \in U\}$ on a ground set $U$ of $n$ elements with joint distribution $p$, a Boolean function $f: 2^U \to \{0, 1\}$, and an (unknown) \emph{scenario} $S = \{e \in U: X_e = 1\}$ of active elements sampled from $p$. We seek to probe the elements one-at-a-time to reveal if they are active until we can certify $f(S) = 1$, while minimizing the expected number of probes. Unlike most previous results that assume independence, we study correlated distributions $p$ and give approximation algorithms for several classes of functions $f$.

    When $f(S)$ is the indicator function for whether $S$ is the spanning set of a given matroid, our problem reduces to finding a basis of active elements of a matroid by probing elements. For \emph{non-adaptive} algorithms, this generalizes several problems in classical combinatorial optimization, including the Min-Sum Set Cover problem (\cite{feige2004approximating, kaplan2005learning}) for $1$-uniform matroids, and the $k$-Min Sum Set Cover problem (\cite{bansal_improved_2020}) for $k$-uniform matroids. We give a non-adaptive $O(\log n)$-approximation algorithm for arbitrary distributions $p$, and show that this is tight up to constants unless P $=$ NP, even for partition matroids. We give a linear programming relaxation that combines knapsack cover inequalities with matroid constraints. 
    Combining techniques from submodular optimization and stochastic optimization, we show that the linear program can be solved in polynomial time. We then use randomized rounding over matroid polytopes to obtain the result. For uniform matroids, we give constant factor $4.642$-approximation that can be further improved to a $2$-approximation if additionally the random variables are negatively correlated for the case of $1$-uniform matroid.
    
    We also give an \emph{adaptive} $O(\log k)$-approximation algorithm for \SBFC for $k$-uniform matroids for the \emph{Graph Probing} problem, where we seek to probe the edges of a graph one-at-a-time until we find $k$ active edges. The underlying distribution on edges arises from (hidden) independent vertex random variables, with an edge being active if at least one of its endpoints is active. This significantly improves over the information-theoretic lower bound on $\Omega(\mathrm{poly}(n))$ (\cite{jiang2019cost}) for adaptive algorithms for $k$-uniform matroids with arbitrary distributions.

\end{abstract}

\newpage
\tableofcontents
\thispagestyle{empty}

\newpage

\setcounter{page}{1}

\section{Introduction}

The \emph{Stochastic Boolean Function Evaluation} (\SBFE) problem is a fundamental problem in stochastic combinatorial optimization, where the objective is to determine the state of a complex system composed of multiple stochastic components. The key assumption is that the overall state of the system can be expressed as a Boolean function over the individual component states. So, in order to determine the state of the system, we must perform tests to reveal the states of these components. One approach to diagnose such systems is to perform tests on all components, which can be prohibitively expensive. A better and more practical approach involves \emph{sequential testing}, where components are tested one by one until the overall state of the system can be determined. Such sequential testing problems arise in a number of applications such as healthcare, manufacturing, telecommunication and project scheduling (see, for example, the surveys by \cite{Moret82} and \cite{UNLUYURT2025356}).

An instance of \SBFE consists of a universe $U$, a joint distribution $p:2^U \to [0,1]$ for the realized active set $S \subseteq U$, and a Boolean function $f:2^U \to \{0,1\}$ describing the system status. Probing an element reveals its state, and the goal is to determine $f(S)$ while minimizing the expected number of probes. Classical examples include the series system (the AND-function)~\citep{B72} and the $k$-of-$n$ problem~\citep{Ben-Dov81}. More generally, \SBFE has been studied for linear threshold functions, score classification, symmetric Boolean functions, and partition matroid constraints~\citep{deshpande2016approximation,gkenosis2018stochastic,ghuge2024nonadaptive,plank2024simple,gkenosis2022stochastic,hellerstein2026approximating}. The problem has also been studied in the machine learning community under the name of active search~\citep{garnett2012bayesian,jiang2017efficient,jiang2019cost}.

Most prior work assumes independent element states, but in many applications these states are correlated, so probing one element may reveal information about others. This makes the problem substantially harder: even $1$-of-$n$ admits only a $4$-approximation under arbitrary correlations~\citep{kaplan2005learning}, and for $k$-of-$n$, no adaptive algorithm using polynomially many oracle calls can achieve an $o(n^{0.16})$-approximation; moreover, the adaptivity gap can be polynomially large~\citep{jiang2019cost}. Beyond these works, \SBFE under correlated distributions remains largely unexplored. This leads to a central question: \emph{what structural assumptions on $f$ or $p$ allow these barriers to be overcome?} 

Under correlated uncertainty, we focus first on \emph{one-sided certification} problems: the realized instance is known to satisfy a monotone feasibility condition, and the goal is to reveal a small witness of this fact using as few probes as possible.

A canonical example of this is in \emph{reliability testing}. A system consists of $n$ components and may already be known to be operational from an inexpensive end-to-end test, even though the states of the individual components remain unknown. Components from the same supplier, batch, or operating environment may exhibit correlated failure modes, so if one component is defective then related components are more likely to be defective as well. By testing components sequentially, the goal is to uncover a short certificate of operability. For a $k$-of-$n$ redundant system, this amounts to identifying $k$ working components. More generally, for a system composed of multiple subsystems, one must verify that each subsystem contains sufficiently many working components, which corresponds to a partition matroid constraint. This is inherently one-sided: the system is already known to work, and the task is only to certify this fact by revealing a feasible set of working components. This motivates the following abstraction.

\medskip
\noindent\textbf{Stochastic Boolean Function Certification Problem (SBFC).}
In this paper, we study a certification variant of stochastic Boolean function evaluation under correlated distributions, which we call \SBFC. The input consists of a universe $U$ of elements, a joint distribution $p:2^U \to [0,1]$ over the realized active set $S \subseteq U$, and a \emph{monotone} Boolean function $f:2^U \to \{0,1\}$. Thus, with probability $p(S)$, exactly the elements in $S$ are active. We assume that every realization in the support of $p$ satisfies $f(S)=1$. A probing strategy sequentially tests elements and observes whether they are active, and it may stop once the observed active elements alone certify feasibility. Formally, if $T \subseteq U$ denotes the set of probed elements, then we require that $f(S \cap T)=1$. The objective is to design a probing strategy that minimizes the expected number of probes. 

Although \SBFC is formally easier than \SBFE, it already captures several important cases. In particular, for the $k$-of-$n$ problem, \SBFE essentially reduces to \SBFC (see \cref{sec: na-k-of-n}). More broadly, the certification viewpoint extends naturally to richer combinatorial settings, including partition constraints and matroids.

More generally, under correlated distributions the difficulty of stochastic probing comes from two distinct sources: the combinatorial complexity of the target function $f$, and the structural complexity of the distribution $p$. Our results are organized along these two complementary axes. 
We study both adaptive and non-adaptive policies for this problem.
A policy is \emph{adaptive} if the choice of the next element to probe may depend on the outcomes observed so far, and \emph{non-adaptive} if elements are probed in a fixed, predetermined order. 

\medskip
\noindent \textbf{Axis I: Structured functions under arbitrary correlations.}
Our first axis asks what is possible when the correlation structure is completely arbitrary, but the target function has additional combinatorial structure. Even for the simplest nontrivial function, the $k$-of-$n$ function, arbitrary correlations already change the picture dramatically: the adaptivity gap can be polynomially large~\citep{jiang2019cost}, in contrast to the independent setting where it is bounded by a constant \cite{nielsen2025non}. Accordingly, meaningful guarantees for non-adaptive algorithms must be stated relative to the optimal non-adaptive benchmark, rather than the adaptive optimum.

Along this axis, we introduce \matroidSBFC, where $f(S)=1$ iff $S$ contains a basis of a given matroid $\mathcal{M}=(U,\mathcal{I})$. This captures $k$-of-$n$ certification (uniform matroid), partition-constrained certification (partition matroid), and spanning-tree certification (graphic matroid) as special cases. We give a tight $\Theta(\log n)$-approximation for non-adaptive matroid-\SBFC under arbitrary distributions, based on an LP relaxation with knapsack-cover inequalities. For the important special case of $k$-uniform matroids, we improve this to a $4.642$-approximation via a connection to the $k$-Min Sum Set Cover problem~\citep{feige2004approximating, bansal_improved_2020}.

\medskip
\noindent \textbf{Axis II: Exploiting correlation structure.}
Our second axis takes the complementary viewpoint: we fix the target function to be $k$-of-$n$ and instead exploit additional structure in the distribution $p$. This direction is motivated by applications such as \emph{drug discovery}, where candidate compounds may share latent biological pathways: if a pathway is active, then several compounds that depend on it are likely to be active as well \cite{kainkaryam2009pooling}. The goal is to identify $k$ active compounds using as few assays as possible. Such dependencies arise from shared latent causes, precisely the kind of structure captured by graph and hypergraph-based probing models.

Concretely, we introduce the \emph{Graph Probing} (\gp) problem: given a graph $G=(V,E)$ with independent Bernoulli vertex variables $\{Y_v\}_{v \in V}$, each edge variable $X_e = Y_u \vee Y_v$ can be probed, and the goal is to find $k$ active edges, if they exist, using as few probes as possible. The observed edge variables are correlated through shared endpoints, even though the latent vertex variables are independent. We give an adaptive $O(\log k)$-approximation for Graph Probing, and complement this with a polynomial lower bound on the adaptivity gap, showing that adaptivity is essential in this model. We further generalize to \emph{Hypergraph Probing}, where each observed variable is the OR of up to $\rho$ independent latent variables, obtaining an $O(\rho^2 \log k)$-approximation.

Additionally, when the distribution satisfies \emph{Conditional Negative Association} (CNA), a natural negative dependence condition, we give a $2$-approximation for non-adaptive $1$-of-$n$ \SBFE, improving the tight $4$-approximation of \cite{kaplan2005learning} for this class of distributions.

\subsection{Problem definitions}\label{sec:problem-definitions}

Having already defined the Stochastic Boolean Function Evaluation (\SBFE) and Stochastic Boolean Function Certification (\SBFC) problem, we now formally define the additional problems studied in this paper. 

\medskip
\noindent{\bf Matroid-\SBFC (\matroidSBFC).} 
We are given a matroid $\mathcal{M} = (U, \mathcal{I})$. Let $\mathcal{S}\subseteq 2^{U}$ denote all the spanning sets of $\mathcal{M}$, i.e., sets that contain a basis. We are given a probability distribution $p:\mathcal{S}\rightarrow [0,1]$ which identifies the joint distribution over the active elements. The goal is to  probe elements until a basis of active elements has been identified, minimizing the expected number of probes. This captures $k$-of-$n$ certification (uniform matroid with rank $k$), partition-constrained certification (partition matroid), and spanning tree certification (graphic matroid) as special cases.

We note that for the special case of $k$-of-$n$ (i.e., $k$-uniform matroids), our results can certify both cases whether the number of active elements is at least $k$ or less than $k$. Certifying $|S| \ge k$ reduces to finding $k$ active elements, while certifying $|S| < k$ is equivalent to finding $n - k + 1$ inactive elements, which is itself a $(n-k+1)$-of-$n$ instance on the complemented distribution. Thus, for $k$-of-$n$, our algorithms solve the $k$-of-$n$-\SBFE problem.

\medskip
\noindent{\bf Graph Probing (\gp).}
The input is a simple undirected graph $G = (V, E)$ with independent Bernoulli vertex variables $\{Y_v\}_{v \in V}$, where $\Pr(Y_v = 1) = p_v$ is known for each $v$. Each edge $e = (u, v) \in E$ has an associated observed variable $X_e = Y_u \lor Y_v$. Only the edge variables can be probed. The goal is to adaptively probe edges until $k$ active edges (with $X_e = 1$) have been found, minimizing the expected number of probes.

\medskip
\noindent{\bf Hypergraph Probing (\hgp).} We also consider a natural generalization of Graph Probing where correlations among observed variables arise from a known hypergraph rather than a graph. Here, we are given a set $L = \{Y_1, \ldots, Y_m\}$ of independent latent Bernoulli variables with $\Pr(Y_j = 1) = p_j$, and a set $R = \{X_1, \ldots, X_n\}$ of observed variables. Each observed variable $X_i$ corresponds to a hyperedge $N(i) \subseteq L$ in a hypergraph $H = (L, R)$, with $X_i = \bigvee_{j \in N(i)} Y_j$. That is, an observed variable is active if and only if at least one of its underlying latent variables is active. Only the observed variables can be probed, and the goal is to find $k$ active observed variables while minimizing the expected number of probes. We define the \emph{rank} $\rho := \max_i |N(i)|$ to be the maximum hyperedge size. Note that Graph Probing is the special case where every hyperedge has size exactly $2$. 

\medskip
\noindent{\bf Conditional Negative Association (CNA).} A distribution $p$ over $2^U$ is said to satisfy CNA if, for any disjoint subsets $A, B \subseteq U$ and any conditioning on the elements outside $A \cup B$, the events $\{e \in S\}$ for $e \in A \cup B$ are negatively associated. Under CNA distributions, we study the $1$-of-$n$ special case of \SBFE.

Throughout, we assume access to the distribution $p$ through either a \emph{sampling oracle} that samples a scenario $S$ from $p$, or through a \emph{joint probability oracle}: given subsets $U_0, U_1 \subseteq U$, this oracle returns the probability $\sum_{S: U_1 \subseteq S, U_0 \cap S = \emptyset} p_S = \Pr_p\left(X_e = 1 \ \forall \ e \in U_1, X_e = 0 \ \forall \ e \in U_0\right)$ that each element in $U_1$ is active and each element in $U_0$ is inactive. These oracles are equivalent up to $\mathrm{poly}(n)$ oracle calls (see \cref{sec: oracle-equivalence}).

\subsection{Results and techniques}\label{sec:results-techniques}

We now state our results formally. 
Our first result is a tight logarithmic approximation for non-adaptive 
\matroidSBFC under arbitrary correlations.

\begin{restatable}{theorem}{NonAdaptiveLogApprox}
    \label{thm: na-log-approximation}
    There is a polynomial-time $O(\log n)$-approximation algorithm for 
    non-adaptive \matroidSBFC for arbitrary distributions. Further, 
    unless $\mathrm{P} = \mathrm{NP}$, there is no $o(\log n)$-approximation algorithm for non-adaptive \matroidSBFC, 
    even for partition matroids.
\end{restatable}

This result is based on first formulating the linear program and applying randomized rounding. The linear program formulation (see \ref{lp: na-strong} in \cref{sec: na-log-upper-bound}) combines \emph{knapsack cover} inequalities along with matroid constraints, leading to a formulation with exponential number of variables as well as constraints. We reformulate the linear program as a convex program where the complexity of the exponential number of variables and constraints is transferred to the objective. While the objective now has exponential number of terms, we show that a random estimate of the gradient can be obtained via parametric submodular function minimization~\cite{iwata2001combinatorial}. Then we can employ standard algorithms for stochastic optimization to solve the program in polynomial time. 

Our rounding algorithm relies on the following sampling result for matroids (\cref{lem: matroid-randomized-rounding}): given a fractional point $y$ in the spanning set polyhedron of matroid, independently rounding every element with probability $\alpha \cdot y$ results in a set with expected rank at least $(1-e^{-\alpha})$ times the rank of the matroid. The sampling result is implicit in \cite{calinescu_maximizing_2007} where it is stated for $\alpha=1$ and we formally state and prove for every $\alpha\geq 1$. Our rounding algorithm rounds the fractional solution after scaling it by $\alpha=O(\log n)$. This allows us to show that the fractional coverage, over  time, of the linear programming solution is matched by the coverage of the integral solution, if allowed $O(\log n)$ factor additional time giving the claimed guarantee.

For the important special case of $k$-of-$n$~\SBFC, we improve the approximation ratio to a constant by utilizing the connection to the $k$-Min Sum Set Cover (\kmssc) problem. In this case, our linear program is essentially the same as one for  \kmssc by~\cite{bansal2010constant}, though, exponential in size. Additionally, we show how to obtain the same guarantee for the general $k$-of-$n$~\SBFE problem as well. 

\begin{restatable}{theorem}{NonAdaptiveKofN}\label{thm: k-of-n-nonadpative-approximation}
    There is a polynomial-time non-adaptive $4.642$-approximation algorithm for both $k$-of-$n$ \SBFC and \SBFE for arbitrary distributions.
\end{restatable}

Designing adaptive algorithms is typically more challenging, and \cite{jiang2019cost} showed that there is no adaptive $o(n^{0.16})$-approximation algorithm for \kCSFE under arbitrary correlations that uses a polynomial number of oracle calls. This raises the natural question: \emph{Can additional structure in the correlation model lead to improved guarantees for adaptive algorithms?} We answer this in the affirmative by giving an adaptive $O(\log k)$-approximation for the Graph Probing model:
\begin{restatable}{theorem}{graphProbing}\label{thm: graph-probing}
    There is a polynomial-time adaptive $O(\log k)$-approximation algorithm for Graph Probing, relative to the optimal adaptive algorithm.
\end{restatable}

Our algorithm exploits the graph structure through a two-way reduction between Graph Probing and a \emph{Vertex Probing} problem. If we could probe vertex $v$ directly, observing $Y_v = 1$ would immediately certify $\deg_v$ active edges, reducing the problem to \emph{Stochastic min-Knapsack} on independent variables, for which an adaptive $3$-approximation is known \citep{deshpande2016approximation}. The reduction from Graph Probing to Vertex Probing is straightforward: each time we probe an edge, we probe both its endpoints. The reverse direction is more delicate. Naively probing all edges incident to a vertex destroys the independence of neighboring vertex variables, so the residual instance is no longer a valid Vertex Probing instance. We resolve this via a recursive probing strategy that maintains a key \emph{independence invariant}: conditioned on all probes so far, the remaining vertex variables are independent. The $O(\log k)$ factor arises from running $O(\log k)$ phases, each halving the residual knapsack requirement.

We further complement this with a lower bound on the adaptivity gap, showing that adaptive algorithms are significantly more efficient in this model.
\begin{restatable}{lemma}{AdaptivityGapGraphProbing}\label{lem: adaptivity-gap-graph-probing}
    The adaptivity gap for Graph Probing on $n$ vertices is $\Omega\left(\frac{\sqrt{n}}{\log n}\right)$.
\end{restatable}

We also consider a natural generalization of Graph Probing to hypergraphs, where each observed variable may depend on up to $\rho$ independent latent variables rather than just two. This captures settings where correlations arise from shared latent causes of higher arity. We obtain the following result for this more general model:

\begin{restatable}{theorem}{hypergraphProbing}
    \label{thm: hypergraph-probing}
    There is a polynomial-time adaptive $O(\rho^2 \log k)$-approximation algorithm for Hypergraph Probing, relative to the optimal adaptive algorithm, where $\rho$ is the rank of the hypergraph.
\end{restatable}

The algorithm generalizes the recursive probing framework from Graph Probing to hypergraphs. The key difference is that probing a single observed variable can now reveal information about up to $\rho$ latent variables, rather than just two. The two factors of $\rho$ in the approximation ratio arise from distinct sources: recovering the independence invariant after each probe requires a recursive procedure whose cost grows by a factor of $\rho$ (due to hyperedges of size $\rho$ branching into $\rho - 1$ further probes in the worst case), and the number of phases increases from $O(\log k)$ to $O(\rho \log k)$ since each phase can only reduce the residual knapsack requirement by a $(1 - 1/\rho)$ factor.

Finally, we consider a different structural assumption on the distribution. When the component states satisfy \emph{Conditional Negative Association} (CNA), a natural and well-studied form of negative dependence, we obtain an improved guarantee for the $1$-of-$n$ problem.

\begin{restatable}{theorem}{CNAApproximation}\label{thm: cna-2-approximation}
    There is a polynomial-time $2$-approximation algorithm for non-adaptive $1$-of-$n$ \SBFE for CNA distributions.
\end{restatable}

This improves the $4$-approximation of \cite{kaplan2005learning}, which holds for arbitrary correlations. Our algorithm is the same natural greedy policy that probes elements in decreasing order of their marginal probability of being active. The improved analysis exploits the CNA property: we construct a charging argument that maps probes of the greedy algorithm to probes of the optimal algorithm, and show that under negative association, the greedy algorithm makes comparable progress.

\section{Related work}\label{sec:related-work}

Stochastic Boolean function evaluation (\SBFE) has a long history in operations research and computer science~\citep{Moret82,UNLUYURT2025356} but almost all work assumes independent random variables. For simple functions like AND/OR, greedy and non-adaptive strategies are optimal~\citep{B72}. The $k$-of-$n$ problem generalizes this, and adaptive optimal algorithms were given by~\cite{Ben-Dov81}. Non-adaptive algorithms are also known: a $2$-approximation by~\cite{gkenosis2018stochastic}, and a tight $1.5$-approximation under uniform costs by~\cite{grammel2022algorithms}.

Approximation algorithms are also known for evaluating certain other functions. In particular, \cite{AllenHK+15} obtained a $(2k)$-approximation for monotone functions in ``disjunctive normal form'' with $k$ terms, and \cite{kaplan2005learning} obtained a logarithmic approximation ratio for functions with {\em both} disjunctive  and conjunctive normal form (CDNF) representations when restricted to the case of unit costs and uniform distributions. Several generalizations have been studied, including for arbitrary costs, \emph{threshold functions}, and with non-Boolean functions such as the \emph{score classification problem} \cite{deshpande2016approximation, JiangLL+20, gkenosis2018stochastic, INZ12,GolovinK-arxiv, plank2024simple,liu2022two,ghuge2024nonadaptive}.

Another closely related problem is the \emph{Stochastic min-Knapsack} problem, where given items with deterministic costs and random rewards, we want to minimize the expected cost to achieve a target total reward value. This problem is closely related to the evaluation of linear threshold functions, and we employ it as a key subroutine in our algorithm for the graph probing problem.

Few results are known under correlated distributions. \cite{kaplan2005learning} gave a $4$-approximation for read-once DNF using a joint probability oracle. \cite{jiang2019cost} also addressed correlated settings in the context of active search, but with different goals. When the support of the probability distribution is polynomial-sized, then non-adaptive $k$-\SBFC (resp. $1$-\SBFC) reduces to the $k$-\mssc (resp. \mssc) problem \cite{feige2004approximating, bansal2010constant, bansal_improved_2020}. This connection was first noticed by \cite{kaplan2005learning}, who generalized the $4$-approximation for \mssc to $1$-\SBFC. Since \mssc is NP-hard to approximate within any constant factor $< 4$, this establishes corresponding hardness lower bounds on $1$-\SBFC. This connection with \mssc has also been exploited for other stochastic problems with correlated distributions, such as the Pandora's Box problem \cite{chawla2021approximating}. The sampling algorithm we use to solve the LP that encodes arbitrary correlations in \cref{sec: na-log-upper-bound} is similar to the algorithm of \cite{swamy2012sampling}, who give approximation algorithms for multistage stochastic optimization problems with correlations.

Lastly, adaptivity gaps and non-adaptive strategies have been studied in many stochastic settings, including knapsack~\citep{dean2008approximating,BGK11}, matching~\citep{BGLMNR12,BehnezhadDH20}, and matroid intersection~\citep{GN13,GuptaNS17}, often yielding approximation guarantees. However, in our setting with correlations, we show that adaptivity gaps can be unbounded.

\section{Non-adaptive \SBFC}\label{sec: na}

We discuss non-adaptive \SBFC in this section. We present the LP relaxation and rounding algorithms in \cref{sec: na-log-upper-bound}, proving the upper bound in \cref{thm: na-log-approximation} and \cref{thm: k-of-n-nonadpative-approximation}. In \cref{sec: na-log-lower-bound}, we prove the hardness lower bound. In \cref{sec: cna}, we prove \cref{thm: cna-2-approximation} for CNA random variables.

\subsection{Randomized rounding algorithms}\label{sec: na-log-upper-bound}

We begin with the linear program formulation, and then give rounding algorithm for arbitrary matroids (\cref{sec: matroid-rounding}) and for $k$-uniform matroids (\cref{sec: na-k-of-n}). For scenario $S$, variable $u_{S, t}$ indicates that $S$ has not been `covered' by the first $t - 1$ probes. Then the objective value when scenario $S$ realizes is exactly $u_{S, \le n} := \sum_{t} u_{S, t}$. Taking expectation over all scenarios gives the objective $\sum_{S} p_S u_{S, \le n}$. 

We have a variable $x_{e, t} \in [0, 1]$ for each element $e \in U$ and time $t \in [n]$ indicating that $e$ is the $t$th element probed. Naturally, we can probe at most $x_t(U) := \sum_{e} x_{e, t} \le 1$ elements at any time $t$.

Let $r: 2^U \to \Z_{\ge 0}$ denote the rank function of the matroid $(U, \mathcal{I})$, and $r^* := r(U)$ denotes the rank of the matroid. Then $u_{S, t} = 0$ only if the first $t - 1$ probes contain a basis $B \subseteq S$, and certainly $x_{< t}(S) := \sum_{e \in S} \sum_{t' < t} x_{e, t} \ge r^*$ in this case. If $x_{< t}(S) < r^*$, then we must have $u_{S, t} = 1$. Combined, we get that $x_{< t}(S) \ge (1 - u_{S, t}) r^*$.

This LP is too weak, and we strengthen it by using knapsack cover-like inequalities:
\begin{align}
    \min & \sum_S p_S u_{S, \le n} & \text{s.t.} \tag{\textup{LP}}\label{lp: na-strong} \\
    x_{< t}(S \setminus T) &\ge (r^* - r(T)) (1 - u_{S, t}), & \forall S \in \mathcal{S}, T \subseteq S, t \in [n], \label{const: knapsack-cover-constraints}\\
    x_t(U) &\le 1 & \forall \ t \in [n], \label{const: schedule-one-element-at-a-time} \\
    x &\ge 0. \notag
\end{align}

This generalizes the LP of \cite{bansal2010constant} for $k$-\mssc. In this problem, we are given a collection $\mathcal{S} = \{S_1, \ldots, S_m\} \subseteq 2^U$ of subsets of size $|S_i| \ge k$ each, and seek to order or \emph{schedule} elements $U$. Subset $S_i$ is uncovered as long as fewer than $k$ elements of $S_i$ have been scheduled, and the objective is to minimize the average covering time across $S_i \in \mathcal{S}$. Choosing the $k$-uniform matroid and the uniform distribution $p$ over $\mathcal{S}$ recovers $k$-\mssc. This connection was first noticed by \cite{kaplan2005learning} for $k = 1$, i.e., for \mssc.

Even for $k$-\mssc, this LP has an exponential number of constraints due to the (necessary) knapsack cover-like constraints \ref{const: knapsack-cover-constraints}.
\cite{bansal2010constant} gave an efficient separation oracle for $k$-\mssc that iterates over all subsets in $\mathcal{S}$. This is inefficient for our LP since the set $\mathcal{S}$ of spanning sets in the matroid is specified implicitly, and can have size \emph{exponential} in $S$. Further, even for a fixed spanning set $S$, we need to separate over constraints over an arbitrary matroid. 

To address this, first, let us bring the variables $u_{S, t}$ from the constraints to the objective, and thus reformulate the LP as a convex optimization problem in just the `assignment' variables $x \in \R^{U \times n}$. For scenario/spanning set $S \in \mathcal{S}$ and time $t \in [n]$, Constraints (\ref{const: knapsack-cover-constraints}) can be rewritten as $u_{S, t} \ge 1 - \frac{x_{< t}(S \setminus T)}{r^* - r(T)}$ for all $T \subseteq S$ that satisfy $r(T) < r^* = r(S)$. An optimal solution $x$ to \ref{lp: na-strong} therefore satisfies
\[
    u_{S, t} = \max\bigg(0, \max_{\substack{{T \subseteq S}: \\ r(T) < r(S)}} \left[1 - \frac{x_{< t}(S) - x_{< t}(T)}{r(S) - r(T)}\right]\bigg).
\]
To make the dependence on $x$ explicit, denote $u_{S, t}(x) := u_{S, t}$, which is convex in $x$. Then \ref{lp: na-strong} can be rewritten as the following convex optimization problem over convex set $\mathcal{K} := \{x \in [0, 1]^{U \times n}: x_{t}(U) \le 1 \ \forall \ t \in [n]\}$:
\begin{align}\label{eqn: lp-strong-reformulated}
    \min f(x) := \min \sum_{t \in [n]} \sum_{S \in \mathcal{S}} p_S u_{S, t}(x) \quad \text{s.t.} \ x \in \mathcal{K}.
\end{align}
There are two issues with applying standard convex optimization algorithms to find the optimal point: (1) not only does the objective $f(x) := \sum_{S \in \mathcal{S}} p_S u_{S, t}(x) \quad \text{s.t.}$ consist of exponentially many terms, and (2) it is also unclear how even a single term can be evaluated for given $x \in \convexbody$. To address (2), we give an efficient `separation oracle' to compute the gradient of a particular term by reducing it to parametrized submodular function minimization \cite{iwata2001combinatorial}. To address (1), we use standard  stochastic optimization techniques and create an efficient, bounded, unbiased estimator for sub-gradients $\partial f(x)$. This leads to the following lemma (see Appendix \ref{sec: omitted-proofs} for the proof):

\begin{lemma}\label{lem: na-lp-poly-time-opt}
    There is a $\mathrm{poly}\left(n, \frac{1}{\varepsilon}\right)$-time algorithm that solves \ref{lp: na-strong} optimally up to additive accuracy $\varepsilon$ for any matroid on $n$ elements.
\end{lemma}

\subsubsection{The $O(\log n)$-approximation for arbitrary matroids}\label{sec: matroid-rounding}

We present the rounding algorithm for \ref{lp: na-strong}.

\begin{algorithm}
    \caption{\textsc{RandomizedRounding}$(x, \alpha)$}\label{alg: randomized-rounding}
    \begin{algorithmic}[1]
        \Statex \textbf{input}: feasible solution $x \in \R^{U \times n}$ to \ref{lp: na-strong}, and parameter $\alpha > 0$
        \For{epoch $\sigma = 1$ to $\lceil \log_2 n \rceil$}
            \For{each $e \in U$ not yet probed} \Comment{In arbitrary order}
                \State toss a coin independently with probability $\min\left(1, \alpha \sum_{t \le 2^\sigma} x_{e, t} \right)$
                \State probe $e$ if coin toss is successful
            \EndFor
        \EndFor
    \end{algorithmic}
\end{algorithm}

\begin{lemma}\label{lem: na-randomized-rounding}
    Algorithm \ref{alg: randomized-rounding} returns a solution with expected cost at most $O(\log n) \cdot \cost(x)$ for $\alpha = 4 \log n$ for solution $x$ to \ref{lp: na-strong} for any instance of \matroidSBFC on $n = |U|$ elements
\end{lemma}

\paragraph{Proof of upper bound in \cref{thm: na-log-approximation}.} \cref{lem: na-lp-poly-time-opt} allows us to compute an optimal solution $x$ to \ref{lp: na-strong} (up to constant additive error $\varepsilon > 0$) in polynomial time. Since \ref{lp: na-strong} is a relaxation of \matroidSBFC, it is sufficient to prove \cref{lem: na-randomized-rounding} to prove \cref{thm: na-log-approximation} for $\alpha = 4 \log n$. We will prove it more generally for all $\alpha \ge 1$.

We proceed to prove \cref{lem: na-randomized-rounding}. Fix a scenario $S \in \mathcal{S}$. In any fractional solution $(x, u)$, $u_{S, t}$ is non-increasing in $t$. Let $t_S$ be the first $t \in [n]$ such that $u_{S, t} < \frac{1}{2}$. We say that $S$ is \emph{half-covered} by the fractional solution in epoch $\sigma$ if $2^{\sigma - 1} < t_S \le 2^{\sigma}$.

By definition, $\sum_{t \in [n]} u_{S, t} \ge \sum_{t \le 2^{\sigma - 1}} u_{S, t} \ge \frac{1}{2} \cdot 2^{\sigma - 1} = \frac{1}{4} 2^{\sigma}$. We claim the following:
\begin{enumerate}
    \item With probability $\ge 1 - \frac{1}{n^2}$, \cref{alg: randomized-rounding} probes at most $(18 \log n) \cdot 2^\sigma$ elements by the end of epoch $\sigma$. This is proven in \cref{lem: na-not-too-many-elements-are-probed}.
    
    \item $S$ is covered by \cref{alg: randomized-rounding} in epoch $\sigma$ with probability $\ge 1 - \frac{1}{n^2}$. This is shown below.
\end{enumerate}

Together, this implies
\begin{align*}
    \left(\text{expected cover time for} \ S \ \text{by Algorithm} \ \ref{alg: randomized-rounding}\right) = O(\log n) \cdot 2^\sigma = O(\log n) \cdot \sum_{t \in [n]} u_{S, t}. 
\end{align*}
\cref{lem: na-randomized-rounding} then follows by linearity of expectation over distribution $p$ on scenarios.

We prove the second claim. Since $u_{S, t} < \frac{1}{2}$ for $t := 2^{\sigma} + 1$. Therefore, by knapsack cover constraints (\cref{const: knapsack-cover-constraints}), we have for all $T \subseteq S$ that
\[
    x_{< t}(S \setminus T) \ge (r^* - r(T)) \cdot \frac{1}{2}.
\]
Define $y_e := 2 x_{e, < t}$ for all $e \in S$. Then, setting $T = \emptyset$, we get $y(S) \ge r^*$. And more generally, $y(S) - y(T) \ge r^* - r(T)$, or that $y(T) \le r(T) + (y(S) - r^*)$, where $y(S) - r^* \ge 0$.

Any element $e \in S$ is probed by the end of epoch $\sigma$ with probability $\ge \min(1, \alpha x_{e, < t}) = \min(1, \frac{\alpha}{2} y_e)$, where $\alpha = 4 \log n$.

Note that $y$ lies in the up-hull of the base polymatroid $\mathcal{B}(\mathcal{M}_S) = \{z \in \R_{\ge 0}^S: z(T) \le r(T) \ \forall \ T \subseteq S, z(S) = r^*\}$: define $z = \frac{r(S)}{y(S)} \cdot y$; it is easy to check that $z \in \mathcal{B}(\mathcal{M})$ and $y \ge z$. We will use the following guarantee for rounding $y$ to a basis.

\begin{restatable}{lemma}{matroidRandomizedRounding}\label{lem: matroid-randomized-rounding}
    Consider a matroid $\mathcal{M}_S = (S, \mathcal{I}_S)$ on ground set $S$ of elements and a point $y$ in the up-hull of the corresponding base polymatroid. For $\beta \ge 1$, denote by $A_\beta \subseteq S$ the random subset of $S$ obtained by including each $e \in S$ in $A_\beta$ independently with probability $\min(1, \beta y_e)$. Then, the expected rank of $A_\beta$ is
    \[
        \E[r(A_\beta)] \ge \left(1 - e^{-\beta}\right) r(S).
    \]
\end{restatable}

We show in \cref{sec: omitted-proofs} that this is a simple extension of the proof of Lemma 5 in \cite{calinescu_maximizing_2007}, who consider the special case when $y$ is exactly in the base polymatroid and $\alpha = 1$. The lemma can also be formulated for arbitrary monotone submodular functions.

By \cref{lem: matroid-randomized-rounding}, the corresponding sampled set $\{e \in S: e \mathrm{\ is \ probed \ by \ the \ end \ of \ epoch \ } \sigma \}$ has expected rank $\ge 1 - \exp(-\alpha/2) = 1 - \frac{1}{n^2}$. Since rank is integral, this sampled set contains a basis with probability $\ge 1 - \exp(-\alpha/2) = 1 - \frac{1}{n^2}$. This proves the claim.

To finish the proof of \cref{thm: na-log-approximation}, it remains to prove the following lemma bounding the number of probes by \cref{alg: randomized-rounding} by the end of epoch $\sigma$:

\begin{restatable}{lemma}{NotTooManyProbes}\label{lem: na-not-too-many-elements-are-probed}
    Fix $\alpha = 4 \log n$. For any $\sigma \in \{1, \ldots, \lceil \log_2 n \rceil\}$, the number of elements probed by the end of epoch $\sigma$ is at most $(18 \log n) \cdot 2^\sigma$ with probability $\ge 1 - \frac{1}{n^2}$.
\end{restatable}

For element $e \in U$, consider random variable $Y_e$ that indicates whether $e$ is probed in epoch $\sigma$. Then $\Pr(Y_e = 1) \le \alpha \cdot \sum_{t \le 2^\sigma} x_{e, t}$, and therefore the number $Y := \sum_{e \in Y} Y_e$ of elements probed in epoch $\sigma$ is $\le \alpha \sum_{e \in U} \sum_{t \le 2^\sigma} x_{e, t} = \alpha \sum_{t \le 2^\sigma} x_t(U) \le \alpha 2^\sigma$ (in expectation). The high probability claim in the lemma then follows by a standard concentration argument using Chernoff bounds (we formalize this argument in \cref{sec: omitted-proofs}).

\subsubsection{The $4.642$-approximation for $k$-uniform matroids}\label{sec: na-k-of-n}

In this section, we give a $4.642$-approximation algorithm for \SBFE on $k$-Uniform matroids (i.e., the $k$-of-$n$ problem). First, we note that it is sufficient to restrict to \SBFC. Indeed, when $k \ge n/2$ and there $\ge k$ elements are active, then any algorithm is a $2$-approximation since the optimal must also probe $\ge n/2$ elements. This is similarly true when $k < n/2$ and fewer than $k$ elements are active: in this case, $\ge n - k + 1 \ge n/2$ elements are inactive, and any algorithm must probe all of them. Thus, we can assume that $k < n/2$, and that at least $k$ elements are active.

Our result generalizes that of \cite{bansal_improved_2020}. Specifically, the \ref{lp: na-strong} reduces to their linear program for $k$-\mssc. Since we can solve \ref{lp: na-strong} to optimality using \cref{lem: na-lp-poly-time-opt}, it remains to round the optimal fractional solution to an integer solution. Like our analysis in \cref{sec: na-log-upper-bound}, \cite{bansal_improved_2020}'s analysis for their kernel-rounding algorithm gives an approximation guarantee for each scenario (i.e., spanning set). Specifically, they show that

\begin{restatable}[\cite{bansal_improved_2020}]{lemma}{KofNScenarioApprox}\label{lem:k-of-n-scenario-approx}
    Given a $k$-uniform matroid and a fractional solution $x$ to \ref{lp: na-strong}, there exists a randomized polynomial-time algorithm for $k$-\SBFC which covers any spanning set $S$ within factor $4.642$ of the cost $\sum_{t} u_{S, t}$ incurred by the fractional solution $x$ in expectation.
\end{restatable}

Along with \cref{lem: na-lp-poly-time-opt} and the observation above, this immediately implies \cref{thm: k-of-n-nonadpative-approximation}.

\subsection{Hardness lower bound}\label{sec: na-log-lower-bound}

We prove the lower bound in \cref{thm: na-log-approximation} in this section by reducing the Hitting Set problem (\texttt{HS}) to non-adaptive \matroidSBFC for partition matroids. In \texttt{HS}, we are given $N$ non-empty subsets $A_1, \ldots, A_N \subseteq U_{\mathtt{HS}}$ of a universe $U_{\mathtt{HS}}$ of $|U_{\mathtt{HS}}| = N$ elements, and seek a subset $S \subseteq U_{\mathtt{HS}}$ of minimum cardinality such that $|S \cap A_i| \ge 1$ for all $i \in [N]$.  \texttt{HS} is equivalent to the Set Cover problem and is NP-hard to approximate within factor $o(\log n)$ \cite{feige_threshold_1998}.\footnote{The number of subsets may not be the same as the number $N$ of elements in general. However, this does not affect hardness and we use this version for clearer exposition.} We prove the following:

\begin{lemma}
    Any $o(\log n)$-approximation algorithm for \PSBFE implies an $o(\log n)$-approximation algorithm for \texttt{HS}.
\end{lemma}

For each instance of \texttt{HS} with $N$ elements, we construct an instance of \PSBFE with a new universe $U$ of $|U| := n = N^2$ elements, and define a partition matroid with $N$ parts on $U$. Then, we show that any $\alpha$-approximation algorithm for this \PSBFE instance can be efficiently converted to an $O(\alpha)$-approximation for the corresponding \texttt{HS} instance. This implies the result.

\paragraph{Reduction.} The new universe $U$ (for \PSBFE) contains $N$ copies of each $e \in U_{\mathtt{HS}}$, with the $\ell$th copy of $e$ denoted $e_{\ell}$ for $\ell \in [N]$. The $N$ parts in the partition matroid are $U_{\ell} := \{e_{\ell}: e \in U_{\mathtt{HS}}\}$ for $\ell \in [N]$, i.e., each part contains one copy of the ground set $U_{\mathtt{HS}}$. We set the requirement $k_\ell = 1$ for each part $1\leq \ell\leq N$.

We define a probability distribution $p: 2^U \to [0, 1]$ on subsets of $U$ through a sampling procedure, described next:
\begin{enumerate}
    \item Choose a permutation $\pi: [N] \to [N]$ uniformly at random. Intuitively, subset $A_{\pi(\ell)}$ is `assigned' to part $U_{\ell}$.
    \item For part $\ell \in [N]$, element $e_{\ell}$ is active if and only if $e \in A_{\pi(\ell)}$.
\end{enumerate}

Suppose we are given an $\alpha$-approximation algorithm for \PSBFE. Given an instance of Hitting Set with elements $U_{\mathtt{HS}}$:
\begin{enumerate}
    \item Construct the instance of \PSBFE described above.
    \item Obtain a solution $\sigma$ to \PSBFE using the $\alpha$-approximation algorithm. Here, $\sigma$ is an ordering of the elements $U = \{e_\ell: e \in U_{\mathtt{HS}}, \ell \in [N]\}$ of the \PSBFE instance.
    \item For each part $\ell \in [N]$, let $\Gamma_\ell$ be the smallest prefix of $\sigma \cap U_\ell$ such that the corresponding elements in $U_{\mathtt{HS}}$ form a hitting set for $\mathtt{ins}_{\mathtt{HS}}$.
    \item Return the hitting set among these, and denote it $\Gamma_{\mathtt{HS}}$.
\end{enumerate}

Denote the original instance of \texttt{HS} as $\mathtt{ins}_{\mathtt{HS}}$, with (deterministic) optimal value $\OPT_{\mathtt{HS}}$. Denote the corresponding \PSBFE instance as $\mathtt{ins}$, with cost of the optimal algorithm denoted $\OPT$, so that the optimal expected value is $\E_p[\OPT]$. Denote the cost of the $\alpha$-approximate algorithm for \PSBFE as $\ALG$; then $\E_p[\ALG] \le \alpha \cdot \E_p[\OPT]$.

\begin{claim}
    $\E_p[\OPT] \le N \cdot \OPT_{\mathtt{HS}}$, so that $\E_p[\ALG] \le \alpha N \cdot  \OPT_{\mathtt{HS}}$.
\end{claim}

\begin{proof}
    Let $S^*_{\mathtt{HS}} \subseteq U_{\mathtt{HS}}$ be the optimal hitting set for the \texttt{HS} instance $\mathtt{ins}_{\mathtt{HS}}$, with optimal $|S^*_{\mathtt{HS}}| = \OPT_0$. Consider the following probing sequence for \texttt{ins}: first, for each $e \in S^*_{\mathtt{HS}}$, probe each of the $N$ copies of $e$ in $U$ in any order. Probe any remaining elements in $U$ in an arbitrary order.

    Fix part $\ell \in [N]$. Irrespective of what set $A_{\pi(\ell)}$ is assigned to part $\ell$ under the random permutation $\pi$, there is some $e \in S^*_{\mathtt{HS}} \cap A_{\ell}$ since $S^*_{\mathtt{HS}}$ is a hitting set for $\mathtt{ins}_{\mathtt{HS}}$. By construction, the copy $e_\ell \in U_\ell$ is active, so that this part is always covered by time $N \cdot |S^*_{\mathtt{HS}}| = N \cdot \OPT_{\mathtt{HS}}$. This holds for all parts with probability $1$.
\end{proof}

Finally, we prove that the solution $\Gamma_{\mathtt{HS}}$ constructed using the reduction above has size $|\Gamma_{\mathtt{HS}}| \le 8 \cdot \frac{\E_p[\ALG]}{N}$, which is $\le 8 \alpha \cdot \OPT_{\mathtt{HS}}$ by the above lemma.

Recall that $|\Gamma_{\mathtt{HS}}| = \min_{\ell \in [N]} |\Gamma_\ell|$ by definition. We claim that with probability $\ge \frac{1}{4}$, we must have $\ALG \ge \frac{N}{2} \min_{\ell \in [N]} |\Gamma_\ell| = \frac{N}{8} |\Gamma_{\mathtt
{HS}}|$, thus implying the result.

For each $\ell \in [N]$, $\Gamma_\ell$ is the smallest prefix of $\ALG$'s ordering $\sigma$ that forms a hitting set for $\mathtt{ins}_{\mathtt{HS}}$. Therefore, there must be some subset $A_i \subseteq U_{\mathtt{HS}}$ that does not intersect any smaller prefix of $\sigma \cap U_\ell$. Denote this set as $A_{\phi(\ell)}$. That is, $\phi: [N] \to [N]$ is some (fixed) mapping of parts in $U$ to subsets in $\mathtt{ins}_{\mathtt{HS}}$.

The probability (for a random permutation $\pi$) that $\phi(\ell) = \pi(\ell)$ is exactly $\frac{1}{N}$. This gives a weak bound $\E_p[\ALG] \ge \frac{1}{N} \cdot (|\Gamma_\ell| - 1) \ge \frac{1}{N} \cdot \frac{|\Gamma_\ell|}{2}$. We strengthen this next.

Denote $\tau = \frac{N}{2} \min_{\ell} |\Gamma_\ell|$, and consider the set $S_{\tau}$ of the first $\tau$ elements in sequence $\sigma$. Let $T \subseteq [N]$ be the collection of parts $\ell$ such that $\Gamma_\ell$ does not lie entirely in $S_\tau$, i.e. $\Gamma_\ell \not\subseteq S_\tau$. Consider some permutation $\pi$. If $\phi(\ell) = \pi(\ell)$ for any $\ell \in T$, then by definition, $S_\tau \cap A_{\phi(\ell)} = \emptyset$, and so none of the elements in $S_\tau \cap U_\ell$ are active. Therefore, $\ALG > \tau$ in this case. We show that this happens with probability $\ge \frac{1}{4}$.

To see this, note first that since each $\Gamma_m$ for $m \not\in T$ is a subset of $S_\tau$, and since sets $\Gamma_m \subseteq U_m$ are in different parts of the partition, we have
\[
   \frac{N}{2} \min_{\ell} |\Gamma_\ell| = \tau \ge \sum_{m \not\in T} |\Gamma_m| \ge (N - |T|) \min_{\ell \in [N]} \Gamma_\ell.
\]
Therefore, $|T| \ge \frac{N}{2}$. The result now follows from the following lemma.

\begin{restatable}{lemma}{badElementsRandomPartition}\label{lem: bad-elements-random-partition}
    Let $N$ be a positive integer and $T \subseteq [N]$ be some subset of $[N]$ of size $|T| \ge N/2$. Let $\phi: [N] \to [N]$ be an arbitrary mapping, and $\pi$ be a permutation on $[N]$ chosen uniformly at random. Then,
    \[
        \Pr\left(\exists \ \ell \in T : \phi(\ell) = \pi(\ell) \right) \ge \frac{1}{4}.
    \]
\end{restatable}

We prove this using the second-moment method in \cref{sec: omitted-proofs}.

\subsection{$2$-Approximation for $1$-CNA}\label{sec: cna}

In this section, we give an algorithm for CNA random variables for $1$-uniform matroids, and present our strategy of the proof of \cref{thm: cna-2-approximation}. The full analysis is deferred to \cref{sec: updated-greedy-full-proof}.

\paragraph{Updated Greedy Algorithm.} A natural algorithm for the \matroidSBFC problem is the \emph{updated} or \emph{conditional greedy} algorithm: it first probes the element $e_1 \in U$ with the highest marginal probability $\Pr_{p}(e_1 \mathrm{\ is \ active})$ of being active, and stops if $e$ is indeed active. If not, it probes the element $e_2$ with the highest marginal probability $\Pr_{p}(e_2 \mathrm{\ is \ active } | e_1 \mathrm{\ is \ not \ active})$ given that $e_1$ is inactive, and so on. Formally, the updated greedy algorithm probes elements $e_1, \ldots, e_n$, defined as follows: having probed $e_1, \ldots, e_{j - 1}$,
\begin{equation}\label{eqn: updated-greedy-rule}
    e_j := {\arg\max}_{e \in U} {\Pr}_{p}(e \mathrm{\ is \ active} | \mathrm{\ none \ of \ } e_1, \ldots, e_{j - 1} \mathrm{\ are \ active}).
\end{equation}

For brevity, we will think of random variables $X \in \{X_e: e \in U\}$ as boolean variables, so that $\Pr(X = 1)$ is denoted as $\Pr(X)$ and $\Pr(X = 0)$ is denoted $\Pr(\overline{X})$. For any $S, T \subseteq U$, the probability $\Pr(X = 0 \ \forall \ X \in S \ \text{and} \ X = 1 \ \forall \ X \in T)$ is denoted as $\Pr(\overline{S}, T)$. \cite{feige2004approximating} showed that this algorithm is a $4$-approximation for $1$-\MSSC and \cite{kaplan2005learning} noted that the algorithm and its approximation guarantee extend to $1$-\SBFE. We improve this to a $2$-approximation when the variables $\{X_e: e \in U\}$ are conditionally negatively associated (CNA). An important consequence of CNA is the following property:

\paragraph{Conditional Negative Association.} We will assume that the random variables $\{X_e: e \in U\}$ satisfy the following \emph{conditional negative cylinder property} or CNCP: denote by $E$ the event that $X = v_X$ for all $X$ in a given subset $S \subseteq U$ for a given realization $v \in \{0, 1\}^S$. Then, for all disjoint $A, B \subseteq U \setminus S$, we must have
\begin{equation}\label{eqn: conditional-negative-cylinder-property}\tag{CNCP}
    \Pr(A, \overline{B} | E) \ge \Pr(A | E) \times \Pr(\overline{B} | E).
\end{equation}

We omit the formal definition of conditional negative association since we will only use \ref{eqn: conditional-negative-cylinder-property} in this work. 

To analyze the algorithm, first, let us relabel the indices so that $X_1, \ldots, X_n$ denotes the probing order of the updated greedy algorithm. For all $j \in [n]$, denote the joint probability $q_j = \Pr(X_j, \overline{X_1}, \ldots, \overline{X_{j - 1}})$, the conditional probability $c_j = \Pr(X_j | \overline{X_1}, \ldots, \overline{X_{j - 1}})$. The `remaining' probability $r_j = \Pr(\overline{X_1}, \ldots, \overline{X_{j - 1}})$ denotes the total probability mass left uncovered at the start of step $j$. Therefore, we have $q_j = c_j r_j$. The cost of the updated greedy algorithm can be written in several ways:
\begin{lemma}
The cost of the updated greedy algorithm is
\begin{align}
    \ugreedy &= \sum_{j \in [n]} j q_j = \sum_{j \in [n]} r_j = \sum_{j \in [n]} \frac{q_j}{c_j}. \label{eqn: updated-greedy-cost}
\end{align}
\end{lemma}
\begin{proof}
    The probability that the algorithm takes exactly $j$ probes is precisely $\Pr(X_j, \overline{X_1}, \ldots, \overline{X_{j - 1}}) = q_j$, and therefore by definition the cost is $\sum_{j \in [n]} j q_j$.
    
    Viewed another way, the algorithm pays a price $1$ after each probe until it finds a nonzero random variable. The price paid at the start of time $j$ is $\Pr(\overline{X_1}, \ldots, \overline{X_{j - 1}}) = r_j$. Therefore, the cost is $\sum_{j \in [n]} r_j$. Finally, $r_j = \frac{q_j}{c_j}$ by definition of conditional expectation, so that the cost can also be written as $\sum_{j \in [n]} \frac{q_j}{c_j}$.
\end{proof}

Let $X_1^*, \ldots, X_n^*$ denote the optimal order. For $t \in [n]$, let $q^*_t = \Pr(X^*_t, \overline{X^*_1}, \ldots, \overline{X^*_{t - 1}})$, $c^*_t = \Pr(X^*_t | \overline{X^*_1}, \ldots, \overline{X^*_{t - 1}})$, and $r^*_t = \Pr(\overline{X^*_1}, \ldots, \overline{X^*_{t - 1}})$. Analogous to \cref{eqn: updated-greedy-cost}, the optimal cost can be written in several ways:
\begin{align}
    \OPT &= \sum_{t \in [n]} t q^*_t = \sum_{t \in [n]} r^*_t = \sum_{t \in [n]} \frac{q^*_t}{c^*_t}. \label{eqn: opt-cost} 
\end{align}

At a high level, the analysis partitions the updated greedy algorithm's probes into several phases of consecutive probes, and maps them to carefully constructed `equivalent' phases in the optimal algorithm. Then, using conditional negative association, we show that the cost incurred by the updated greedy algorithm in each phase does not exceed a \emph{proxy cost} incurred by the optimal in the corresponding phase, using \cref{eqn: opt-cost} above. The proxy cost is designed so that it is always within factor $2$ of the true cost, giving the $2$-approximation. This is formalized in \cref{sec: cna-general-case}.

\section{Adaptive approximation for Graph Probing}\label{sec: graph-probing}

In this section, we consider the Graph Probing (\gp) problem, where we are given a graph $G = (V, E)$, knapsack size $1 \le k \le |E|$, and probabilities $\Pr(Y_v = 1)$ for independent Bernoulli random variables $Y_v$ on each vertex $v \in V$. These induce random variables $X_{uv} = Y_u \lor Y_v$ on each edge $uv \in E$, i.e., $X_{uv} = 1$ if and only if at least one of $Y_u$ and $Y_v$ is $1$. The goal is to (adaptively) probe the edges in some order until $k$ nonzero edges are found, or until all edges have been probed.

We give an $O(\log k)$-approximation for \gp (Theorem \ref{thm: graph-probing}), improving upon the $\Omega(n^{0.16})$ lower bound on approximations for adaptive algorithms for \kCSFE (\cite{jiang2019cost}).
Before proving the main result of the section, we first show that adaptivity is necessary for \gp by giving an instance where an adaptive algorithm can do significantly better than a non-adaptive one, by a polynomial in $n = |V|$ factor.

\AdaptivityGapGraphProbing*

The proof relies on the following idea: for any star graph where the leaf vertices are always inactive and only the central vertex can be active, there are only two outcomes: either all edges are active or all edges are inactive. While an adaptive algorithm can probe a single edge to determine the outcome for all other edges, a non-adaptive algorithm might have to probe all edges. We choose our graph as the union of many star graphs; the detailed proof can be found in \cref{app: graph-probing-omitted-proofs}.

This motivates us to consider adaptive algorithms for \gp. Our main result gives an adaptive algorithm that has approximation $O(\log k)$ with respect to any adaptive algorithm:

\graphProbing*

\subsection{Algorithm sketch}

Our main idea is to reduce this to the problem of probing vertices: we consider a different problem, where instead of probing edges, we are required to probe the vertices. Upon probing vertex $v$, if the realized value of the corresponding random variable $Y_v$ is $1$, we get reward $\deg_v$, and $0$ otherwise. Our goal is to (adaptively) probe the vertices in some order until our total reward is $\ge k$. We call this problem \emph{Vertex Probing}. Vertex Probing is equivalent to the Stochastic min-Knapsack problem \citep{deshpande2016approximation}, where given (unrealized) independent Bernoulli random variables $Y_v, v \in V$ with corresponding sizes $s_v \ge 0$ and some knapsack size $k > 0$, we seek to probe the random variables until we probe set $S \subseteq V$ with $\sum_{v \in S} s_v Y_v \ge k$, or until each random variable has been probed.

Our graph probing algorithm probes edges with some `guidance' from an algorithm for Vertex Probing. Our approach consists of three broad steps, described next. We give the formal algorithm and the full analysis in \cref{app: graph-probing-omitted-proofs}.

\paragraph{Vertex Probing subroutine.} Since the vertex variables $Y_v$ are independent, Vertex Probing is the Stochastic min-Knapsack problem, with the size equal to the degree. There exists an adaptive $3$-approximation algorithm due to \cite{deshpande2016approximation} for Stochastic min-Knapsack, and therefore for Vertex Probing. We treat this algorithm as a subroutine.
    
\paragraph{Graph Probing to Vertex Probing.} We show in \cref{sec: graph-probing-to-vertex-probing} that any algorithm for \gp on graph $G$ with knapsack size $k$ can be converted to an algorithm for Vertex Probing on the same graph $G$ and knapsack size $k$, while at most doubling the cost (number of probes):

\begin{restatable}{lemma}{GraphProbingToVertexProbing}\label{lem: transformation-graph-probing-to-vertex-probing}
    Given a graph $G$, knapsack size $k$, and an adaptive algorithm for \gp with cost $N$, there exists an adaptive algorithm for Vertex Probing with cost $\le 2N$.
\end{restatable}

In particular, the optimal value for Vertex Probing is at most twice the optimal value for \gp. Our transformation is simple: each time the algorithm for \gp seeks to probe some edge $uv$, the corresponding Vertex Probing algorithm probes both $u$ and $v$ (unless they have already been probed).
    
\paragraph{Vertex Probing to Graph Probing.} Next, in \cref{sec: vertex-probing-to-graph-probing}, we give an algorithm (\cref{alg: vertex-probing-to-graph-probing}) that converts an algorithm for Vertex Probing with knapsack size $k$ to an algorithm for \gp with knapsack size $k/2$, while increasing the number of probes by at most $k$.

The main idea behind our algorithm is as follows: each time the algorithm for Vertex Probing probes a vertex $v$, the corresponding \gp algorithm starts probing edges incident on $v$, until it finds a $0$ edge. But the algorithm cannot stop now and defer to the vertex probing algorithm to find the next vertex to probe since the neighboring vertices are no longer independent in the graph probing instance. To resolve this issue, the algorithm keeps recursively probing all edges at a vertex where all previous probes at incident edges have resulted in $1$ edges. We call this subroutine \textsc{ProbeVertex} and formalize it as \cref{alg: probe-vertex}.

This ensures that, conditioned on the probes done so far, the residual instance of graph probing has independent random variables at the vertices:

\begin{restatable}{lemma}{ResidualInstance}\label{lem: graph-probing-residual-instance-is-independent-main}
    Consider graph $G = (V, E)$ where edges $F \subseteq E$ have been probed and the residual graph is $G' = (V', E')$. At the end of each run of \textsc{ProbeVertex}, the vertex random variables $\{Y_v: v \in V'\}$ are independent of the outcomes of the probed edges $F$.
\end{restatable}

Following this idea, we get the following algorithm for \gp: given graph $G$ and knapsack size $k$, run \cite{deshpande2016approximation}'s algorithm for Vertex Probing with knapsack size $k$. This is called the first \emph{phase} of the algorithm. While fewer than $k$ non-zero edges have been found, re-run this on the residual instance. By the above observations, the number of probes of this algorithm is guaranteed to be within a constant factor of the number of probes of the optimal adaptive algorithm for \gp.

Next, we use the transformation described in (3) to convert this Vertex Probing algorithm into a \gp algorithm, with the guarantee that at least $k/2$ non-zero edges are found. That is, at least half of the knapsack is filled in this iteration:
\begin{restatable}{lemma}{HalfBagSize}\label{lem: graph-probing-half-bag-size}
    At the end of the first phase, \cref{alg: vertex-probing-to-graph-probing} has either terminated or has $k' \le \frac{k}{2}$.
\end{restatable}

Further, in this phase, \cref{alg: vertex-probing-to-graph-probing} uses fewer probes than the optimal uses overall:

\begin{lemma}\label{lem: graph-probing-first-phase-simple}
    The number of probes in the first phase of Algorithm~\ref{alg: vertex-probing-to-graph-probing} is at most $8 \ \OPT^e(G, k)$, where $\OPT^e(G, k)$ is the cost of the optimal algorithm for the graph probing instance.
\end{lemma}

We repeat this process on the remaining instance until $k$ nonzero edges are found (or we have probed all edges). Since we halve the required number of edges to fill the knapsack in each iteration, this algorithm takes at most $1 + \log_2 k$ iterations. Since our cost is within a constant factor of the cost of the optimal \gp algorithm at each iteration, our total cost is within factor $O(\log_2 k)$ of the optimal.

These above lemmas together imply the theorem:

\begin{proof}[Proof of Theorem \ref{thm: graph-probing}]
    We will show that Algorithm \textsc{VertexProbingToGraphProbing} on input graph $G$ and knapsack size $k$ has cost at most $(8 + 8 \log_2 k) \ \OPT^e(G, k)$.
    Let $F$ be the (random) set of edges probed in phase 1. Let $G' = (V', E')$ be the graph remaining at the end of the phase, and $k'$ be the remaining knapsack size, i.e., there were exactly $k - k'$ nonzero edges in $F$. 
    
    If $k' = 0$, then by Lemma \ref{lem: graph-probing-first-phase-simple}, we have $\ALG^e(G, k) \le 8 \ \OPT^e(G, k)$. If $k' \neq 0$ and $G'$ is empty, then the graph has fewer than $k$ nonzero edges. Therefore, $\ALG^e(G, k) = \OPT^e(G, k) = |E|$.
    Otherwise, by \Cref{lem: graph-probing-half-bag-size}, we have $0 < k' \le \frac{k}{2}$. 
    Moreover, by \Cref{lem: graph-probing-residual-instance-is-independent-main}, the residual instance with graph $G' = (V', E')$ and knapsack size $k'$ is another instance of \gp. Then, by induction on the number of phases, we have
    \(\ALG^e(G', k') \le (8 + 8 \log_2 k') \ \OPT^e(G', k').\)
    Further, since there are exactly $k - k'$ nonzero edges in $F$, $\OPT^e(G, k)$ must find at least $k'$ nonzero edges in the residual instance, so that
    \(\OPT^e(G, k) \ge k' + \OPT^e(G', k') \ge \OPT^e(G', k').\)
    Therefore, by Lemma \ref{lem: graph-probing-first-phase-simple} and induction on the number of phases, we have
    \begin{align*}
        \ALG^e(G, k) \,\, &\le \,\, 8 \ \OPT^e(G, k) + \ALG^e(G', k') \,\, \le \,\, 8 \ \OPT^e(G, k) + (8 + 8 \log_2 k') \ \OPT^e(G', k') \\
        &\le \,\, 8 \ \OPT^e(G, k) + (8 + 8 \log_2 k') \ \OPT^e(G, k) \,\, \le \,\, 8(1 + \log_2 k) \ \OPT^e(G, k),
    \end{align*}
    where the final inequality holds since $k' \le k/2$.
\end{proof}

\subsection{Hypergraph Probing}\label{sec: hypergraph-probing}

In this section, we extend the adaptive algorithm for Graph Probing (Section~\ref{sec: graph-probing}) to the Hypergraph Probing (\hgp) problem. Recall that in \hgp, we have independent latent Bernoulli variables $L = \{Y_1, \ldots, Y_m\}$ with $\Pr(Y_j = 1) = p_j$, observed variables $R = \{X_1, \ldots, X_n\}$ with $X_i = \bigvee_{j \in N(i)} Y_j$, and rank $\rho := \max_i |N(i)|$. The goal is to probe observed variables until $k$ active ones are found, minimizing the expected number of probes. 
Note that Graph Probing is the special case $\rho = 2$.

The algorithm and analysis are direct extensions of the ideas from Graph Probing. The main structural difference is that each observed variable depends on up to $\rho$ latent variables (instead of $2$), which affects the analysis in two places: the branching factor in the recursion tree and the rate at which the residual knapsack size decreases per phase.

\paragraph{Overview of the algorithm.} The algorithm follows the same overall structure with the following natural generalizations:
\begin{itemize}
    \item \textbf{Latent Variable Probing (\textsc{LVP}).} We define a \emph{Latent Variable Probing} problem: probe latent variables $Y_j$ adaptively until $\sum_{j \in S} w_j Y_j \ge k$, where $w_j := |\{i \in [n] : j \in N(i)\}|$ is the degree of $Y_j$ in the hypergraph. Since the latent variables are independent, this is also an instance of Stochastic min-Knapsack, admitting a $3$-approximation~\citep{deshpande2016approximation}.

    \item \textbf{From \hgp to \textsc{LVP}.} Next, we show that any \hgp algorithm with cost $N$ can be converted to an \textsc{LVP} algorithm with cost at most $\rho \cdot N$: when the \hgp algorithm probes $X_i$, we probe all latent variables in $N(i)$; this generalizes Lemma~\ref{lem: transformation-graph-probing-to-vertex-probing}.

    \item \textbf{From \textsc{LVP} to \hgp.} Finally, we generalize the vertex probing algorithm. 
    When the \textsc{LVP} algorithm directs us to probe $Y_j$, we probe edges incident to $j$ one by one. 
    If $X_i = 0$, all latent variables in $N(i)$ are $0$. If $X_i = 1$, we recursively call the algorithm on each affected latent variable in $N(i) \setminus \{j\}$ to restore independence. 
    This generalizes Algorithm~\ref{alg: probe-vertex}; the key difference is that each $1$-probe can affect up to $\rho - 1$ latent variables. We will refer to this algorithm as \textsc{ProbeLatent}.

    \item \textbf{Phases.} The algorithm proceeds in phases, re-initializing the $3$-approximation for \textsc{LVP} on the residual instance in each phase.
\end{itemize}

\paragraph{Analysis.}
The analysis parallels the analysis for Graph Probing in \cref{sec: graph-probing}. We state the key lemmas, highlighting where the rank $\rho$ enters.

\begin{lemma}\label{lem: hgp-independence-invariant}
At the end of each run of \textsc{ProbeLatent}, the latent variables $\{Y_j : j \in L'\}$ in the residual hypergraph $H' = (L' \cup R', F')$ are independent of the outcomes of all probed observed variables.
\end{lemma}
The proof is identical in structure to Lemma~\ref{lem: graph-probing-residual-instance-is-independent-main}: \textsc{ProbeLatent} ensures that for any $j \in L'$, no observed variable $X_i$ with $j \in N(i)$ has been probed, and independence follows since the latent variables are originally independent.

\begin{lemma}\label{lem: hgp-zeros-vs-ones}
    Let $m_1$ and $m_0$ denote the number of non-zero and zero observed variables probed during a single run of \textsc{ProbeLatent}. Then $m_0 \le (\rho - 1) m_1 + 1$.
\end{lemma}

\noindent \textit{Proof Sketch.} Consider the recursion tree $\tau$ rooted at the initial call \textsc{ProbeLatent}$(j)$, where each node corresponds to a latent variable on which \textsc{ProbeLatent} is called. Each $1$-probe at a node $j'$ creates at most $\rho - 1$ children (one per affected latent variable in $N(i) \setminus \{j'\}$), while each $0$-probe creates none and terminates the direct probing at $j'$. Since each node contributes at most one $0$-probe, $m_0 \le |V_\tau|$. The number of edges satisfies $|A_\tau| \le (\rho-1) m_1$, and since $\tau$ is a tree, $m_0 \le |V_\tau| = |A_\tau| + 1 \le (\rho-1) m_1 + 1$. \qed 

The next two lemmas bound the cost of each phase with respect to the total cost of an optimal algorithm. Their proofs parallel the corresponding Graph Probing lemmas, with the factor-$2$ blowup replaced by $\rho$ throughout; we omit the details.

\begin{lemma}\label{lem: hgp-probelatent-calls}
    The number of calls to \textsc{ProbeLatent} in the first phase is at most $3\rho \cdot \OPT(H, k)$.
\end{lemma}

\begin{lemma}\label{lem: hgp-first-phase}
    The number of probes in the first phase is at most $O(\rho) \cdot \OPT(H, k)$.
\end{lemma}

Moreover, we can guarantee that at least $k/\rho$ non-zero edges are found in the first phase. 

\begin{lemma}\label{lem: hgp-knapsack-reduction}
    At the end of the first phase, the algorithm has either terminated or $k' \le k\left(1 - \frac{1}{\rho}\right)$.
\end{lemma}

Combining these lemmas proves Theorem~\ref{thm: hypergraph-probing}: each phase costs $O(\rho) \cdot \OPT(H, k)$ (Lemma~\ref{lem: hgp-first-phase}), the residual instance remains a valid \hgp instance by Lemma~\ref{lem: hgp-independence-invariant}, and the knapsack size decreases by a $(1 - 1/\rho)$ factor per phase (Lemma~\ref{lem: hgp-knapsack-reduction}), so $O(\rho \log k)$ phases suffice. The total cost is $O(\rho) \cdot O(\rho \log k) \cdot \OPT(H, k) = O(\rho^2 \log k) \cdot \OPT(H, k)$.

\bibliographystyle{alpha}

\begin{thebibliography}{CGMT21}

\bibitem[AHKU17]{AllenHK+15}
S.~Allen, L.~Hellerstein, D.~Kletenik, and T.~\"{U}nl\"{u}yurt.
\newblock Evaluation of monotone {DNF} formulas.
\newblock {\em Algorithmica}, 77:661--685, 2017.

\bibitem[AS16]{alon2016probabilistic}
N.~Alon and J.H.~Spencer.
\newblock {\em The Probabilistic Method}.
\newblock 2016.

\bibitem[BBFT20]{bansal_improved_2020}
N.~Bansal, J.~Batra, M.~Farhadi, and P.~Tetali.
\newblock Improved approximations for min sum vertex cover and generalized min
  sum set cover.
\newblock arXiv:2007.09172, 2020.

\bibitem[BD81]{Ben-Dov81}
Y.~Ben-Dov.
\newblock Optimal testing procedures for special structures of coherent
  systems.
\newblock {\em Management Science}, 27(12):1410--1420, 1981.

\bibitem[BDH20]{BehnezhadDH20}
S.~Behnezhad, M.~Derakhshan, and M.T.~Hajiaghayi.
\newblock Stochastic matching with few queries: $(1-\epsilon)$ approximation.
\newblock Symposium on Theory of Computing ({STOC}), 2020.

\bibitem[BGK10]{bansal2010constant}
N.~Bansal, A.~Gupta, and R.~Krishnaswamy.
\newblock A constant factor approximation algorithm for generalized min-sum set
  cover.
\newblock Symposium on Discrete Algorithms ({SODA}), 2010.

\bibitem[BGK11]{BGK11}
A.~Bhalgat, A.~Goel, and S.~Khanna.
\newblock Improved approximation results for stochastic knapsack problems.
\newblock Symposium on Discrete Algorithms ({SODA}), 2011.

\bibitem[BGL{\etalchar{+}}12]{BGLMNR12}
N.~Bansal, A.~Gupta, J.~Li, J.~Mestre, V.~Nagarajan, and A.~Rudra.
\newblock When {LP} is the cure for your matching woes: Improved bounds for
  stochastic matchings.
\newblock {\em Algorithmica}, 63(4):733--762, 2012.

\bibitem[Bub15]{bubeck2015convex}
S.~Bubeck.
\newblock Convex optimization: Algorithms and complexity.
\newblock {\em Foundations and Trends in Machine Learning}, 8(3-4):231--357,
  2015.

\bibitem[But72]{B72}
R.~Butterworth.
\newblock Some reliability fault-testing models.
\newblock {\em Operations Research}, 20(2):335--343, 1972.

\bibitem[CCPV07]{calinescu_maximizing_2007}
G.~Calinescu, C.~Chekuri, M.~P\'{a}l, and J.~Vondr\'{a}k.
\newblock Maximizing a submodular set function subject to a matroid constraint.
\newblock Integer Programming and Combinatorial Optimization ({IPCO}), 2007.

\bibitem[CGMT21]{chawla2021approximating}
S.~Chawla, E.~Gergatsouli, J.~McMahan, and C.~Tzamos.
\newblock Approximating {P}andora's box with correlations.
\newblock arXiv:2108.12976, 2021.

\bibitem[Che52]{chernoff1952measure}
H.~Chernoff.
\newblock A measure of asymptotic efficiency for tests of a hypothesis based on
  the sum of observations.
\newblock {\em Annals of Mathematical Statistics}, pages 493--507, 1952.

\bibitem[DGV08]{dean2008approximating}
B.C.~Dean, M.X.~Goemans, and J.~Vondr\'{a}k.
\newblock Approximating the stochastic knapsack problem: The benefit of
  adaptivity.
\newblock {\em Mathematics of Operations Research}, 33(4):945--964, 2008.

\bibitem[DHK16]{deshpande2016approximation}
A.~Deshpande, L.~Hellerstein, and D.~Kletenik.
\newblock Approximation algorithms for stochastic submodular set cover with
  applications to boolean function evaluation and min-knapsack.
\newblock {\em ACM Transactions on Algorithms}, 12(3):1--28, 2016.

\bibitem[Fei98]{feige_threshold_1998}
U.~Feige.
\newblock A threshold of $\ln n$ for approximating set cover.
\newblock {\em Journal of the ACM}, 45(4):634--652, 1998.

\bibitem[FLT04]{feige2004approximating}
U.~Feige, L.~Lov\'{a}sz, and P.~Tetali.
\newblock Approximating min sum set cover.
\newblock {\em Algorithmica}, 40:219--234, 2004.

\bibitem[GGHK18]{gkenosis2018stochastic}
D.~Gkenosis, N.~Grammel, L.~Hellerstein, and D.~Kletenik.
\newblock The stochastic score classification problem.
\newblock European Symposium on Algorithms ({ESA}), 2018.

\bibitem[GGHK22]{gkenosis2022stochastic}
D.~Gkenosis, N.~Grammel, L.~Hellerstein, and D.~Kletenik.
\newblock The stochastic boolean function evaluation problem for symmetric
  boolean functions.
\newblock {\em Discrete Applied Mathematics}, 309:269--277, 2022.

\bibitem[GGN24]{ghuge2024nonadaptive}
R.~Ghuge, A.~Gupta, and V.~Nagarajan.
\newblock Nonadaptive stochastic score classification and explainable
  half-space evaluation.
\newblock {\em Operations Research}, 2024.

\bibitem[GHKL22]{grammel2022algorithms}
N.~Grammel, L.~Hellerstein, D.~Kletenik, and N.~Liu.
\newblock Algorithms for the unit-cost stochastic score classification
  problem.
\newblock {\em Algorithmica}, 84(10):3054--3074, 2022.

\bibitem[GK17]{GolovinK-arxiv}
D.~Golovin and A.~Krause.
\newblock Adaptive submodularity: A new approach to active learning and
  stochastic optimization.
\newblock arXiv:1003.3967, 2017.

\bibitem[GKX{\etalchar{+}}12]{garnett2012bayesian}
R.~Garnett, Y.~Krishnamurthy, X.~Xiong, J.~Schneider, and R.~Mann.
\newblock Bayesian optimal active search and surveying.
\newblock International Conference on Machine Learning ({ICML}), 2012.

\bibitem[GN13]{GN13}
A.~Gupta and V.~Nagarajan.
\newblock A stochastic probing problem with applications.
\newblock Integer Programming and Combinatorial Optimization ({IPCO}), 2013.

\bibitem[GNS17]{GuptaNS17}
A.~Gupta, V.~Nagarajan, and S.~Singla.
\newblock Adaptivity gaps for stochastic probing: Submodular and {XOS}
  functions.
\newblock Symposium on Discrete Algorithms ({SODA}), 2017.

\bibitem[HPS26]{hellerstein2026approximating}
L.~Hellerstein, B.~Plank, and K.~Schewior.
\newblock Approximating matroid basis testing on partition matroids using
  budget-in-expectation.
\newblock Symposium on Discrete Algorithms ({SODA}), 2026.

\bibitem[IFF01]{iwata2001combinatorial}
S.~Iwata, L.~Fleischer, and S.~Fujishige.
\newblock A combinatorial strongly polynomial algorithm for minimizing
  submodular functions.
\newblock {\em Journal of the ACM}, 48(4):761--777, 2001.

\bibitem[INvdZ16]{INZ12}
S.~Im, V.~Nagarajan, and R.~van~der~Zwaan.
\newblock Minimum latency submodular cover.
\newblock {\em ACM Transactions on Algorithms}, 13(1):13:1--13:28, 2016.

\bibitem[JGM19]{jiang2019cost}
S.~Jiang, R.~Garnett, and B.~Moseley.
\newblock Cost effective active search.
\newblock {\em Advances in Neural Information Processing Systems}, 32, 2019.

\bibitem[JLLS20]{JiangLL+20}
H.~Jiang, J.~Li, D.~Liu, and S.~Singla.
\newblock Algorithms and adaptivity gaps for stochastic $k$-{TSP}.
\newblock Innovations in Theoretical Computer Science ({ITCS}), 2020.

\bibitem[JMC{\etalchar{+}}17]{jiang2017efficient}
S.~Jiang, G.~Malkomes, G.~Converse, A.~Shofner, B.~Moseley, and R.~Garnett.
\newblock Efficient nonmyopic active search.
\newblock International Conference on Machine Learning ({ICML}), 2017.

\bibitem[KKM05]{kaplan2005learning}
H.~Kaplan, E.~Kushilevitz, and Y.~Mansour.
\newblock Learning with attribute costs.
\newblock Symposium on Theory of Computing ({STOC}), 2005.

\bibitem[KW09]{kainkaryam2009pooling}
R.M.~Kainkaryam and P.J.~Woolf.
\newblock Pooling in high-throughput drug screening.
\newblock {\em Current Opinion in Drug Discovery \& Development}, 12(3):339,
  2009.

\bibitem[Liu22]{liu2022two}
N.~Liu.
\newblock Two 6-approximation algorithms for the stochastic score
  classification problem.
\newblock arXiv:2212.02370, 2022.

\bibitem[Mor82]{Moret82}
B.M.E.~Moret.
\newblock Decision trees and diagrams.
\newblock {\em ACM Computing Surveys}, 14(4):593--623, 1982.

\bibitem[NRS25]{nielsen2025non}
M.A.~Nielsen, L.~Rohwedder, and K.~Schewior.
\newblock Non-adaptive evaluation of $k$-of-$n$ functions: Tight gap and a
  unit-cost {PTAS}.
\newblock arXiv:2507.05877, 2025.

\bibitem[PS24]{plank2024simple}
B.M.~Plank and K.~Schewior.
\newblock Simple algorithms for stochastic score classification with small
  approximation ratios.
\newblock {\em SIAM Journal on Discrete Mathematics}, 38(3):2069--2088, 2024.

\bibitem[PZ30]{paley1930some}
R.E.A.C.~Paley and A.~Zygmund.
\newblock On some series of functions (1).
\newblock {\em Mathematical Proceedings of the Cambridge Philosophical
  Society}, 26:337--357, 1930.

\bibitem[SS12]{swamy2012sampling}
C.~Swamy and D.B.~Shmoys.
\newblock Sampling-based approximation algorithms for multistage stochastic
  optimization.
\newblock {\em SIAM Journal on Computing}, 41(4):975--1004, 2012.

\bibitem[{\"{U}}n25]{UNLUYURT2025356}
T.~\"{U}nl\"{u}yurt.
\newblock Sequential testing problem: A follow-up review.
\newblock {\em Discrete Applied Mathematics}, 377:356--369, 2025.

\end{thebibliography}

\newcommand{\etalchar}[1]{$^{#1}$}

\newpage
\appendix

\section{Equivalence of sampling and joint probability oracles}\label{sec: oracle-equivalence}

Here, we show that the sampling and joint probability oracles are equivalent in terms of number of samples, up to a polynomial factor.
Recall that given subsets $U_0, U_1 \subseteq U$, a \emph{joint probability oracle} returns the probability $\sum_{S \supseteq U_1, S \cap U_0 = \emptyset} p(S) = \Pr_p(X_e = 1 \ \forall \ e \in U_1, X_e = 0 \ \forall \ e \in U_0)$ returns the probability that all elements in $U_1$ are active and all elements in $U_0$ are inactive. A sampling oracle, on the other hand, samples scenarios $S$ from distribution $p$.

Suppose first that we are given a joint probability oracle. To obtain a sampling oracle, order ground set elements $e_1, \ldots, e_n$, and sample the sequence of sets $\emptyset = S_0 \subseteq S_1 \subseteq \ldots \subseteq S_n = S$, where $S_t$ is obtained by adding $e_t$ to $S_{t - 1}$ with probability $\Pr_p(X_{e_t} = 1 | X_e = 1 \ \forall \ e \in S_{t - 1}, X_e = 0 \ \forall e \in \{e_1, \ldots, e_{t - 1}\} - S_{t - 1})$. It is easy to see that this generates independent samples from distribution $p$.
Suppose now that we are given a sampling oracle. Then, for any $\varepsilon > 0$, we can get an estimate for $\sum_{S \supseteq U_1, S \cap U_0 = \emptyset} p(S) = \Pr_p(X_e = 1 \ \forall \ e \in U_1, X_e = 0 \ \forall \ e \in U_0)$ by sampling $\mathrm{poly}(n, 1/\varepsilon)$ independent samples, using standard concentration inequalities.

\section{Omitted proofs from \cref{sec: na}}\label{sec: omitted-proofs}

We provide proofs omitted from \cref{sec: na} here.

\subsection{Proof of \cref{lem: matroid-randomized-rounding}}

We prove the matroid randomized rounding lemma next:

\matroidRandomizedRounding*

\begin{proof}
    Denote $q_e := \min(1, \beta y_e)$ for $e \in S$. Since $r(A_\beta)$ is non-decreasing in $y$ and $y$ is in the up-hull of the base polymatroid $\mathcal{B}(\mathcal{M}_S)$, it is sufficient to prove the lemma when $y \in \mathcal{B}(\mathcal{M}_S)$. We show that the proof of Lemma 5 in \cite{calinescu_maximizing_2007} generalizes to our setting, who prove this for $\beta = 1$.
    
    To make the dependence on $y$ explicit, let us denote $R(y) := \E[r(A_\beta)]$. By linearity of expectation and independence,
    \[
        R(y) = \E[r(A_\beta)] = \sum_{T \subseteq S} \bigg(\prod_{e \in T} q_e \prod_{e \in S \setminus T} (1 - q_e) \bigg) r(T).
    \]
    For $e \in S$ and $T \subseteq S$, denote $r_e(T) := r(T \cup \{e\}) - r(T)$. Now define function $s: \R^S \to \R$ as
    \[
        s(y) := \min_{T \subseteq S} \Big( r(T) + \sum_{e \in S \setminus T} y_e r_e(T) \Big).
    \]
    We will show that $R(y) \ge \left(1 - e^{-\beta}\right) s(y) \ge \left(1 - e^{-\beta}\right) r(S)$.
    
    Since $s$ is a minimum of linear functions, it is concave in $y$. In particular, since $\mathcal{P}(\mathcal{M}_S)$ is a polytope, $s$ attains its minimum over $\mathcal{P}(\mathcal{M}_S)$ at some vertex. Since vertices of a base polymatroid are indicator vectors of its bases, we get that $s(y) \ge \min_{y' \in \mathcal{P}(\mathcal{M}_S)} s(y') = r(S)$.

    Therefore, it remains to prove that $R(y) \ge \left(1 - e^{-\beta}\right) s(y)$ for all $y$. As in \cite{calinescu_maximizing_2007}, define a Poisson clock $\mathcal{C}_e$ of rate $y_e$ for each $e \in S$ that starts at time $t = 0$, and emits a signal in any infinitesimal interval $dt$ with probability $y_j dt$, independently of other intervals. We maintain a set $B_t$ and include $e \in B_t$ as soon as a signal it emitted by $\mathcal{C}_e$. We stop the clocks at $t = \beta$, instead of $t = 1$ as in their case.

    The random event $\{e \in B_\beta\}$ is equivalent to sampling $e$ independently with probability $1 - \exp(- \beta y_e) \le \min(1, \beta y_e) = q_e$. Therefore, $\E[r(B_\beta)] \le \E[r(A_\beta)]$. The rest of the analysis shows that $\E[r(B_\beta)] \ge (1 - \exp(-\beta)) s(y)$ is essentially the same, and we omit it here.
\end{proof}

\subsection{Proof of \cref{lem: na-not-too-many-elements-are-probed}}

We prove \cref{lem: na-not-too-many-elements-are-probed} here:

\NotTooManyProbes*

\begin{proof}
    For element $e \in U$, consider random variable $Y_e$ that indicates whether $e$ is probed in epoch $\sigma$. Then $\Pr(Y_e = 1) \le \alpha \cdot \sum_{t \le 2^\sigma} x_{e, t}$, and therefore the number $Y := \sum_{e \in Y} Y_e$ of elements probed in epoch $\sigma$ is $\le \alpha \sum_{e \in U} \sum_{t \le 2^\sigma} x_{e, t} = \alpha \sum_{t \le 2^\sigma} x_t(U) \le \alpha 2^\sigma$ (in expectation).

    Since elements are probed independently, we can apply Chernoff bound (\cite{chernoff1952measure}) to $Y = \sum_{e \in U} Y_e$, which says that for any upper bound $\mu \ge \E Y$, we have
    \[
        \Pr(Y \ge (1 + \delta) \ \mu) \le \exp\left(-\frac{\delta^2 \ \mu}{2 + \delta}\right) 
    \]
    Choose $\delta = 2$ and $\mu = 3 \log n + \E[Y]$, so that
    \begin{align*}
        \Pr\left(Y > 9 \log n + 3 \ \E [Y] \right) &\le \exp\left(-\frac{4}{3} (3 \log n + \E [Y])\right) \le \frac{1}{n^3}.
    \end{align*}
    Therefore, $Y \le 9 \log n + 3 \ \E [Y] \le (9 \log n) \cdot 2^\sigma$ with probability $\ge 1 - \frac{1}{n^3}$. Taking a union bound over epochs $\tau = 1, \ldots, \sigma$ gives that the number of elements probed by the end of epoch $\sigma$ is at most $(18 \log n) \cdot 2^\sigma$ with probability $\ge 1 - \frac{1}{n^2}$.
\end{proof}

\subsection{Proof of \cref{lem: na-lp-poly-time-opt}}\label{sec: solving-na-lp}

In this section, we show that \ref{lp: na-strong} can be solved efficiently and prove \cref{lem: na-lp-poly-time-opt}.

Recall that we reformulate the LP as the following optimization problem in \cref{eqn: lp-strong-reformulated}:
\begin{align*}
    \min f(x) = \min \sum_{t \in [n]} \sum_{S \in \mathcal{S}} p_S u_{S, t}(x) \quad \text{s.t.} \ x \in \mathcal{K}.
\end{align*}
Here, $u_{S, t}(x) := \max\bigg(0, \max_{\substack{{T \subseteq S}: \\ r(T) < r(S)}} \left[1 - \frac{x_{< t}(S) - x_{< t}(T)}{r(S) - r(T)}\right]\bigg)$, and the convex set $\mathcal{K} := \{x \in [0, 1]^{U \times n}: x_{t}(U) \le 1 \ \forall \ t \in [n]\}$.

For any $S \in \mathcal{S}$ and any $t$, $u_{S, t}(x)$ is a maximum of linear functions in $S$, and therefore convex. Therefore the objective function $f(x) := \sum_{t \in [n]} \sum_{S \in \mathcal{S}} p_S u_{S, t}(x)$ is convex.

Unfortunately, $f$ has exponentially many terms, and we cannot efficiently compute either the objective or its gradient exactly. It is unclear how even to efficiently compute even a single term $u_{S, t}(x)$. We give an efficient and unbiased estimator for the gradient by generalizing the separation oracle described in \cite{bansal2010constant} to matroids using parameterized submodular function minimization (\cref{lem: compute-subgradient-for-scenario}). With these simplifications, we can use stochastic gradient descent to obtain the optimal solution:

\begin{lemma}[\cite{bubeck2015convex}, Theorem 6.1]\label{lem: sgd}
    Let $f$ be a convex function on a convex body $\convexbody \subseteq \mathbb{R}^d$ with diameter
    $D = \max_{x,y \in \convexbody} \|x - y\|_2$. Let $G$ be an unbiased estimator for a sub-gradient of $f$ that satisfies $\|G(x)\|_2 \le B$ with probability $1$ for all $x \in \convexbody$. For any $\varepsilon > 0$, projected SGD converges to $x \in \convexbody$ with
    \[
      \E [f(x)] - \min_{x' \in K} f(x') \le \varepsilon
    \]
    using at most $T =  \frac{2 D^2 B^2}{\varepsilon^2}$ oracle calls to $G$ and projections on $\convexbody$.
\end{lemma}

The diameter $D$ of $\mathcal{K}$ is at most
\[
    D \le \sqrt{\sum_{t \in [n]} \sum_{e \in U} x^2_{e, t}} \le \sqrt{\sum_{t \in [n]} \sum_{e \in U} x_{e, t}} \le \sqrt{n \times 1} = \sqrt{n}.
\]
Further, it is easy to see that projection on $\mathcal{K}$ is efficient.

We will use \cref{lem: sgd} to minimize $f$ over $\mathcal{K}$. To complete the proof of \cref{lem: na-lp-poly-time-opt}, it is therefore sufficient to give an unbiased estimator $G$ for a sub-gradient $\partial f(x)$ that (1) can be computed in a polynomial number of oracle calls to the sampling oracle for distribution $p$, and (2) has norm $\|G(x)\|_2$ upper bounded by $\mathrm{poly}(n)$ for all $x \in \convexbody$ with probability $1$. We describe $G$ explicitly next, and show that both claims follow.

\paragraph{Computing sub-gradient $\partial u_{S, t}(x)$ for given scenario $S \in \mathcal{S}$ and point $x \in \convexbody$.} Given some point $x \in \convexbody$, scenario $S \in \mathcal{S}$, and time $t \in [n]$, we show that $\partial u_{S, t}(x)$ can be computed efficiently, using a technique similar to the separation oracle constructed in \cite{bansal2010constant}. Since $u_{S, t}(x)$ is a piecewise linear function in $x$, the sub-gradient is not uniquely defined at all points. However, as is standard, we can use the gradient of any of the linear pieces at the given point for the sub-gradient at the point.

\begin{lemma}\label{lem: compute-subgradient-for-scenario}
    Let $S \in \mathcal{S}$, $t \in [n]$, and $x \in \convexbody$. Define set $\hat{T}_{S, t}(x) = \hat{T}$ as follows:
    \[
        \hat{T} := {\arg\max}_{T \subseteq S: r(T) < r(S)} \left[1 - \frac{x(S) - x(T)}{r(S) - r(T)}\right].
    \]
    Then, (1) $\hat{T}$ can be computed efficiently, and (2) for $e \in U$ and $t' \in [n]$,
    \[
        \frac{\partial u_{S, t}(x)}{\partial x_{e, t'}} = \begin{cases} -\frac{1}{r(S) - r(\hat{T})} & \text{if} \ u_{S, t}(x) \neq 0, t' < t, \ \mathrm{and} \ e \in S \setminus \hat{T}\\
        0 & \text{otherwise}.
        \end{cases}
    \]
\end{lemma}

\begin{proof}
    (1) Fix $S, t, x$, and denote $y_e := x_{e, < t}$. Then $\hat{T} = {\arg\min}_{T \subseteq S} \frac{y(S) - y(T)}{r(S) - r(T)}$. For $\lambda \ge 0$, define set function $h_\lambda: 2^S \to \R$ as
    \[
        h_\lambda(T) := \left(y(S) - y(T)\right) - \lambda \left(r(S) - r(T)\right) = \left(\lambda r(T) - y(T)\right) + \left(y(S) - \lambda r(S)\right).
    \]
    Then ${\min}_{T \subseteq S} \frac{y(S) - x(T)}{r(S) - r(T)} \ge \lambda$ if and only if $h_\lambda(T) \ge 0$. Since $h_\lambda$ is a submodular function for all $\lambda$, we can find the minimizer $\hat{T}$ be searching over $\lambda$ and minimizing $h_\lambda$ over subsets of $S$ \cite{iwata2001combinatorial}. This proves the first part.

    (2) If $u_{S, t}(x) = 0$, then the sub-gradient $\partial u_{S, t}(x) = 0$ since $u_{S, t}$ is a piecewise linear function.
    
    Otherwise, suppose $u_{S, t}(x) \neq 0$. Then $u_{S, t}(x) = \left(1 - \frac{y(S) - y(\hat{T})}{r(S) - r(\hat{T})}\right)$, and for given $e \in U$ and $t' \in [n]$,
    \[
        \frac{\partial u_{S, t}(x)}{\partial x_{e, t'}} = \frac{\partial u_{S, t}}{\partial y_{e}} \cdot \frac{\partial y_e}{\partial x_{e, t'}}.
    \]
    The first factor $\frac{\partial u_{S, t}}{\partial y_{e}} = - \frac{1}{r(S) - r(\hat{T})}$ if $e \in S \setminus \hat{T}$, and $0$ otherwise. The second factor $\frac{\partial y_e}{\partial x_{e, t'}} = 1$ if $t' < t$ and $0$ otherwise. The result follows.
\end{proof}

The above lemma allows us to construct an unbiased gradient estimator $G(x) = \partial f(x)$ efficiently, as follows:
\begin{enumerate}
    \item Sample scenario $S \sim p$ from distribution $p$ on scenarios $\mathcal{S}$.
    \item For each $t \in [n]$, compute $\hat{T}_{S, t}(x)$ using the above lemma.
    \item Return estimate $G(x) = \sum_{t \in [n]} \partial u_{S, t}(x)$ using the above lemma.
\end{enumerate}
Then, using linearity of expectation and sub-gradients, we get
\[
    \E_{S \sim p}[G(x)] = \sum_{S} p_S \sum_{t \in [n]} \partial u_{S, t}(x) = \partial \left(\sum_{t} \sum_{S} p_S u_{S, t}(x)\right) = \partial f(x).
\]

\paragraph{Norm bound on gradient estimate.} Finally, we upper bound the norm of $G(x)$, completing the proof.
\begin{lemma}\label{lem: gradient-estimate-bound}
    $\|G(x)\|_2 \le n^2$ for all $x \in \convexbody$, with probability $1$.
\end{lemma}

\begin{proof}
    Fix scenario $S \in \mathcal{S}$, time $t \in [n]$, and $x \in \convexbody$. By the previous lemma, for all $e \in U$ and $t' \in [n]$, $\left|\frac{\partial u_{S, t}(x)}{\partial x_{e, t'
    }}\right| \le 1$. Summing across time $t \in [n]$,
    \[
        \left| \frac{\partial \left(\sum_{t} u_{S, t}(x)\right)}{\partial x_{e, t'}}\right| \le n,
    \]
    and therefore, in this scenario, $\|G(x)\|_2 \le n \cdot \sqrt{|U| \times n} = n^2$.
\end{proof}

\subsection{Proof of \cref{lem: bad-elements-random-partition}}

We prove the following lemma:

\badElementsRandomPartition*

\begin{proof}
Assume without loss of generality that $|T| = N/2$, since increasing the size can only increase the probability of the event stated in the lemma.

For each $\ell \in T$, define the indicator random variable
\[
    Z_\ell = \mathbf{1}\!\bigl[\phi(\ell) = \pi(\ell)\bigr],
\]
and let $Z = \sum_{\ell \in T} Z_\ell$. Since $\pi$ is a uniformly random permutation, $\pi(\ell)$ is uniform over $[N]$ for each $\ell$, so $\E[Z_\ell] = \frac{1}{N}$, and so $\E[Z] = 1/2$.

We now compute the second moment of $Z$ and apply the Paley-Zygmund inequality (\cite{paley1930some}; see also \cite{alon2016probabilistic}) to $Z$. 
\[
    \E[Z^2] = \sum_{\ell \in T} \E[Z_\ell] + \sum_{\substack{\ell, m \in T: \\ \ell \neq m}} \E[Z_\ell Z_m].
\]
For $\ell \neq m$, the event $\{Z_\ell\, Z_m = 1\}$ requires $\pi(\ell) = \phi(\ell)$ and $\pi(m) = \phi(m)$ simultaneously. Since $\pi$ is a permutation, $\pi(\ell) \neq \pi(m)$. If $\phi(\ell) = \phi(m)$, then $Z_\ell Z_m = 0$. Otherwise, $\E[Z_\ell Z_m] = \Pr(Z_\ell Z_m = 1) = \frac{1}{N(N - 1)}$. Therefore, $\E[Z^2] \le |T| \cdot \frac{1}{N} + |T|(|T| - 1) \cdot \frac{1}{N(N - 1)} < 1$.

Since $Z$ is an integer random variable, the Paley-Zygmund inequality gives
\[
    \Pr(Z \ge 1) = \Pr(Z > 0)
    \ge \frac{\bigl(\E[Z]\bigr)^2}{\E[Z^2]} > \frac{1}{4}. \qedhere
\]
\end{proof}

\section{Proof of \cref{thm: cna-2-approximation}}\label{sec: updated-greedy-full-proof}

Here, we prove \cref{thm: cna-2-approximation}. It might be helpful to go through Section \cref{sec: motivating-example-updated-greedy} first where we prove the theorem for a special case, highlighting the main steps in the proof. However, \cref{sec: cna-general-case} can be read independently of \cref{sec: motivating-example-updated-greedy}.

\subsection{A motivating example}\label{sec: motivating-example-updated-greedy}

First, we analyze the updated greedy algorithm for a special case (stated next), making several assumptions that we remove later. Let $\OPT$ cover everything by time $K \in [n]$, that is, remaining mass $r^*_{K + 1} = \Pr(\overline{X^*_1}, \ldots, \overline{X^*_K}) = 0$. Let $j_1, \ldots, j_K$ denote the indices such that $X_{j_i} = X^*_{i}$ for $i \in [K]$. Denote $j_0 = 0$. Suppose $X^*_1, \ldots, X^*_K$ is a subsequence of $X_1, \ldots, X_n$, that is, $j_1 < j_2 < \ldots < j_K$. Intuitively, the updated greedy algorithm follows the optimal order but uses some extra random variables. \emph{The key idea is to (roughly) charge the cost $\sum_{j = j_{i - 1} + 1}^{j_i} r_j$ incurred by the updated greedy algorithm between times $j_{i - 1}$ and $j_{i}$ to the cost $r^*_i$ incurred by the optimal at time $i$.}

Since $X^*_1, \ldots, X^*_K$ cover everything, and since $X^*_i = X_{j_i}$, the updated greedy algorithm covers everything by time $j_K$, that is, $q_{j_{K} + 1} = \ldots = q_{n} = 0$. Therefore, from eqn. (\ref{eqn: updated-greedy-cost}),
\[
    \ugreedy = \sum_{j \in [n]} \frac{q_j}{c_j} = \sum_{j \in [j_K]} \frac{q_j}{c_j} = \sum_{i = 1}^K \sum_{j = j_{i - 1} + 1}^{j_i} \frac{q_j}{c_j}.
\]
As we will show in our key Lemma \ref{lem: conditional-probability-relation-greedy-and-opt-2}, for $j_{i - 1} < j \le j_{i}$, the conditional probability $c_j \ge c^*_{i}$. Intuitively, the `rate' at which updated greedy makes progress between $j_{i - 1}$ and $j_{i}$ is at least $c^*_i$. Therefore, $\ugreedy$ is upper bounded by:
\[
    \ugreedy = \sum_{i = 1}^K \sum_{j = j_{i - 1} + 1}^{j_i} \frac{q_j}{c_j} \le \sum_{i = 1}^K \sum_{j = j_{i - 1} + 1}^{j_i} \frac{q_j}{c^*_i} = \sum_{i = 1}^K \frac{1}{c^*_i} \sum_{j = j_{i - 1} + 1}^{j_i} q_j.
\]
Further, since $r_j$ denotes the probability mass uncovered at the start of step $j$, we have $\sum_{j = j_{i - 1} + 1}^{j_i} q_j = \sum_{j = j_{i - 1} + 1}^n q_j - \sum_{j = j_{i} + 1}^n q_j = r_{j_{i - 1} + 1} - r_{j_{i} + 1}$. Therefore,
\begin{equation}\label{eqn: updated-greedy-upper-bound-1}
    \ugreedy \le \sum_{i = 1}^K \frac{1}{c^*_i} (r_{j_{i - 1} + 1} - r_{j_i + 1}) = \sum_{i = 1}^K r_{j_{i - 1} + 1} \left(\frac{1}{c^*_i} - \frac{1}{c^*_{i - 1}} \right).
\end{equation}

Next, since $\{X^*_1, \ldots, X^*_{i - 1}\} \subseteq \{X_1, \ldots, X_{j_{i -1}}\}$, we have $r_{j_{i - 1} + 1} \le r^*_{i}$. This would help us get an upper bound on $\ugreedy$ if we had $\frac{1}{c^*_i} - \frac{1}{c^*_{i - 1}} \ge 0$, i.e., if $c^*_i$ was monotone decreasing in $i$. Unfortunately, this is false in general: for example, whenever $K > 1$, we have $c^*_K = 1 > q_1^* = c^*_1$. However, as we show later in Lemma \ref{lem: monotone-conditional-probability}, we can replace conditional probabilities $c^*_i$ with `proxy' values $d^*_i \le c^*_i$ that satisfy $d^*_{1} \ge \ldots \ge d_n^*$, and with the proxy cost within factor $2$ of the optimal cost $\OPT$:
\[
    \OPT_{\mathrm{proxy}} := \sum_{i \in [K]} \frac{q^*_i}{d^*_i} \le 2 \sum_{i \in [K]} \frac{q^*_i}{c^*_i} = 2 \ \OPT.
\]
Let's go back to the upper bound (\ref{eqn: updated-greedy-upper-bound-1}) on $\ugreedy$ and replace $c^*_{i}$ with the proxy $d^*_i$. Then, let's use our earlier observation $r_{j_{i - 1} + 1} \le r^*_{i}$ and monotonicity $d^*_{i - 1} \ge d^*_{i}$, giving us the following:
\begin{align*}
    \ugreedy &\le \sum_{i = 1}^K \frac{1}{c^*_i} (r_{j_{i - 1} + 1} - r_{j_{i} + 1}) & (\mathrm{\cref{eqn: updated-greedy-upper-bound-1}}), \\
    &\le \sum_{i = 1}^K \frac{1}{d_i^*} (r_{j_{i - 1} + 1} - r_{j_{i} + 1}) & (d_i^* \le c_i^*), \\
    &= \sum_{i = 1}^K r_{j_{i - 1} + 1} \left(\frac{1}{d_{i}^*} - \frac{1}{d_{i - 1}^*}\right) & (\text{rearranging the sum}), \\
    &\le \sum_{i = 1}^K r^*_{i} \left(\frac{1}{d_{i}^*} - \frac{1}{d_{i - 1}^*}\right) & (r_{j_{i - 1} + 1} \le r^*_i), \\
    &= \sum_{i = 1}^K \frac{1}{d^*_i}(r^*_{i} - r^*_{i + 1}) & (\text{rearranging the sum}), \\
    &= \sum_{i = 1}^K \frac{q^*_{i}}{d_i^*} = \OPT_{\text{proxy}} \le 2 \ \OPT.
\end{align*}

\subsection{The general case}\label{sec: cna-general-case}

Next, we will remove the assumption that $X^*_1, \ldots, X^*_K$ is a subsequence of $X_1, \ldots, X_n$. Recall that $X_1, \ldots, X_n$ is the probing order of the updated greedy algorithm, and $X^*_1, \ldots, X^*_n$ is the optimal probing order. The outline of the proof is as follows: we will define indices $j_0 < j_1 < \ldots < j_K$ for the updated greedy algorithm and indices $t_0 < t_1 < \ldots < t_K < t_{K + 1}$ for the optimal. The goal is to (roughly) charge the cost $\sum_{j = j_{i - 1} + 1}^{j_i} r_j$ incurred by the updated greedy algorithm between $j_{i - 1}$ and $j_{i}$ to the cost $\sum_{t = t_{i}}^{t_{i + 1} - 1} r^*_t$ incurred by the optimal between $t_{i}$ and $t_{i + 1}$.

We begin by defining \emph{thresholds}, which map each index $j \in [n]$ in the updated greedy algorithm to a `threshold index' $T(j) \in [n]$ in the optimal:

\begin{definition}[Threshold index]\label{def: thereshold-index}
    Given $j \in [n]$, define $T(j)$ as the smallest index $t \in [n]$ such that $X^*_t \not\in \{X_1, \ldots, X_{j - 1}\}$. $T(j)$ is called the threshold index for $j$.
\end{definition}

Intuitively, threshold indices allow us to lower bound the progress made by the updated greedy algorithm: our next lemma relates the conditional probability $c_j = \Pr(X_j | \overline{X_1}, \ldots, \overline{X_{j - 1}})$ with the conditional probability $c^*_{T(j)} = \Pr(X^*_{T(j)} | \overline{X^*_1}, \ldots, \overline{X^*_{T(j) - 1}})$. \emph{This lemma is the only place where we use conditional negative association.}

\begin{lemma}\label{lem: conditional-probability-relation-greedy-and-opt-1}
    Suppose random variables $X_1, \ldots, X_n$ satisfy the conditional negative cylinder property \ref{eqn: conditional-negative-cylinder-property}. Then, for all $j \in [n]$, $c^*_{T(j)} \le c_j$.
\end{lemma}

\begin{proof}
    Denote the set $A = \{X_1, \ldots, X_{t - 1}\}$. By definition, $X^*_{T(j)} \not\in A$ and $X^*_t \in A$ for all $t < T(j)$. Denote $B = \{X^*_1, \ldots, X^*_{T(j) - 1}\}$, and $C = A \setminus B$.

    By the greedy algorithm update rule (eqn. (\ref{eqn: updated-greedy-rule})),
    \begin{align*}
        c_j &= \Pr(X_j | \overline{X_1}, \ldots, \overline{X_{j - 1}})  = \Pr(X_j | \overline{A}) \\
        &\ge \Pr(X^*_{T(j)} | \overline{A}) \\
        &= \Pr(X^*_{T(j)}, \overline{A}) \times \frac{1}{\Pr(\overline{A})} \\
        &= \Pr(X^*_{T(j)}, \overline{B}, \overline{C}) \times \frac{1}{\Pr(\overline{A})} \\
        &= \Pr(X^*_{T(j)}, \overline{C} | \overline{B}) \times \Pr(\overline{B}) \times \frac{1}{\Pr(\overline{A})}.
    \end{align*}
    This sets us up to use \ref{eqn: conditional-negative-cylinder-property}: since $X^*_{T(j)} \not\in B$ and $X^*_{T(j)} \not\in C$, we get
    \begin{align*}
        c_j &\ge {\Pr(X^*_{T(j)}, \overline{C} | \overline{B})} \times \Pr(\overline{B}) \times \frac{1}{\Pr(\overline{A})} \\
        &{\ge \Pr(X^*_{T(j)} | \overline{B}) \times \Pr(\overline{C} | \overline{B})} \times \Pr(\overline{B}) \times \frac{1}{\Pr(\overline{A})} & {(\ref{eqn: conditional-negative-cylinder-property})},\\
        &= \Pr(X^*_{T(j)} | \overline{B}) & (\text{since }\Pr(\overline{B}, \overline{C}) = \Pr(\overline{A})), \\
        &= c^*_{T(j)}.
    \end{align*}
\end{proof}

Next, we define indices $j_0, j_1, \ldots$ and $t_0, t_1, \ldots$, starting with $t_0 = j_0 = 0$ and $t_1 = 1$. Having defined $t_i$, define $j_i$ as the index $j$ such that $X^*_{t_i} = X_j$, and then define $t_{i + 1} = T(j_i + 1)$. Stop when $j_{i + 1} = n$ to get the sequence $j_0, j_1, \ldots, j_K = n$. It follows from the definitions that $j_0 < j_1 < \ldots < j_K$ and $t_0 < t_1 < \ldots < t_K$. $t_{K + 1}$ is defined as $n + 1$ for brevity.

We illustrate this with an example: suppose $n = 8$, with $X^*_1 = X_4$, $X_2^* = X_7$, $X_3^* = X_5$, and $X_4^* = X_8$. Then, $t_1 = 1$, and $j_1 = 4$. $t_2 = T(j_1 + 1) = T(5) = 2$, and $j_2 = 7$ since $X^*_{t_2} = X^*_2 = X_7$. $t_3 = T(j_2 + 1) = T(8) = 4$, and $j_3 = 8 = n$. We stop since $j_4 = n$, with $K = 4$ and $t_5$ defined as $n + 1 = 9$.

Note that for all $j \in [j_{i - 1} + 1, j_{i}]$, we have $T(j) = t_{i}$. The following lemma then follows from Lemma \ref{lem: conditional-probability-relation-greedy-and-opt-1}:

\begin{lemma}\label{lem: conditional-probability-relation-greedy-and-opt-2}
    For $j_{i - 1} < j \le j_{i}$, the conditional probability $c_j \ge c^*_{t_{i}}$.
\end{lemma}

This helps us bound the cost of the updated greedy algorithm:
\begin{align*}
    \ugreedy &= \sum_{j \in [n]} \frac{q_j}{c_j} = \sum_{j \in [j_K]} \frac{q_j}{c_j} = \sum_{i = 1}^K \sum_{j = j_{i - 1} + 1}^{j_i} \frac{q_j}{c_j} \\
    &\le \sum_{i = 1}^K \sum_{j = j_{i - 1} + 1}^{j_i} \frac{q_j}{c^*_{t_{i}}} & (\text{Lemma}\ \ref{lem: conditional-probability-relation-greedy-and-opt-2}) \\
    &= \sum_{i = 1}^K \frac{1}{c^*_{t_{i}}} \sum_{j = j_{i - 1} + 1}^{j_i} q_j.
\end{align*}

Further, since $r_j$ denotes the mass uncovered by the updated greedy algorithm at the start of step $j$, we have $\sum_{j = j_{i - 1} + 1}^{j_i} q_j = \sum_{j = j_{i - 1} + 1}^n q_j - \sum_{j = j_{i} + 1}^n q_j = r_{j_{i - 1} + 1} - r_{j_{i} + 1}$. Therefore,
\begin{equation}\label{eqn: updated-greedy-upper-bound-2}
    \ugreedy \le \sum_{i = 1}^K \frac{1}{c^*_{t_{i}}} (r_{j_{i - 1} + 1} - r_{j_i + 1}) = \sum_{i = 1}^K r_{j_{i - 1} + 1} \left(\frac{1}{c^*_{t_i}} - \frac{1}{c^*_{t_{i - 1}}} \right).
\end{equation}

The next lemma compares the progress made by optimal with the progress made by the greedy algorithm:

\begin{lemma}\label{lem: updated-greedy-measuring-progress}
    For all $i \ge 1$, we have $r^*_{t_i} \ge r_{j_{i - 1} + 1}$.
\end{lemma}

\begin{proof}
    By definition, $t_i = T(j_{i - 1} + 1)$ is the smallest index $t$ such that $X^*_t \not\in \{X_1, \ldots, X_{j_{i - 1}}\}$. Therefore, for all $t < t_i$, $X^*_t \in \{X_1, \ldots, X_{j_{i - 1}}\}$. Consequently, $r^*_{t_i} = \Pr(\overline{X^*_1}, \ldots, \overline{X^*_{t_{i} - 1}}) \ge \Pr(\overline{X_1}, \ldots, \overline{X_{j_{i - 1}}}) \allowbreak = r_{j_{i - 1} + 1}$.
\end{proof}

Therefore, if we had $c^*_{t_{i - 1}} \ge c^*_{t_{i}}$ for all $i$, then (\cref{eqn: updated-greedy-upper-bound-2}) would give the upper bound $\sum_{i = 1}^K r^*_{t_{i}} \left(\frac{1}{c^*_{t_i}} - \frac{1}{c^*_{t_{i - 1}}}\right) \allowbreak = \sum_{i = 1}^K \frac{r^*_{t_i} - r^*_{t_{i + 1}}}{c^*_{t_i}}$ on $\ugreedy$. However, as we previously noted, this is false: for example, $c^*_n = 1$ but $c^*_{t_i} < 1$ for all $i \le K$. Our final lemma shows that we can replace conditional probabilities $c^*_{t}$ with `proxy' values $d^*_t \le c^*_t$ such that $d^*_1 \ge d^*_2 \ge \ldots \ge d^*_n > 0$, with the proxy optimal cost $\OPT_{\mathrm{proxy}} := \sum_{t \in [n]} \frac{q_t^*}{d_t^*}$ within a factor $2$ of $\OPT$:

\begin{lemma}\label{lem: monotone-conditional-probability}
    There exist numbers $1 \ge d_1^* \ge \ldots \ge d^*_n > 0$ such that $d^*_t \le c^*_t$ for all $t \in [n]$, and such that the proxy optimal cost $\OPT_{\mathrm{proxy}} := \sum_{t \in [n]} \frac{q_t^*}{d_t^*}$ satisfies
    \[
        \OPT = \sum_{t \in [n]} \frac{q^*_t}{c^*_t} \ge \frac{1}{2} \sum_{t \in [n]} \frac{q_t^*}{d_t^*} = \frac{1}{2} \ \OPT_{\mathrm{proxy}}.
    \]
\end{lemma}

\begin{proof}[Proof of Lemma \ref{lem: monotone-conditional-probability}]
    Define $d^*_t = \min_{s \le t} c^*_t$, so that $d^*_t \le c^*_t$ for all $t$ and $d^*_1 \ge \ldots \ge d^*_n = \min_{t \in [n]} c^*_t > 0$.

    To bound the cost $\OPT_{\mathrm{proxy}}$, we first prove the following claim: for all $s \le t$, $\frac{1}{c^*_s} \le \frac{1}{c^*_t} + (t - s)$.

    We prove the claim by induction on $t - s$. First, note that $q_1^* \ge \ldots \ge q_n^*$, since otherwise, there exists some $t$ with $q^*_t < q^*_{t + 1}$ and we can improve the objective value $\OPT = \sum_{t} t q^*_t$ by swapping $X^*_t$ and $X^*_{t + 1}$.

    Now, noting that $r_s^* = \frac{q_s^*}{c_s^*}$ and $r_{s + 1}^* = \frac{q_{s + 1}^*}{c_{s + 1}^*}$, we have
    \[
        q_s^* = r_s^* - r_{s + 1}^* = \frac{q_s^*}{c_s^*} - \frac{q_{s + 1}^*}{c_{s + 1}^*} \ge q_s^* \left(\frac{1}{c_s^*} - \frac{1}{c_{s + 1}^*}\right),
    \]
    where the inequality follows since $q_s^* \ge q_{s + 1}^*$. Rearranging, $\frac{1}{c_s^*} \le \frac{1}{c_{s + 1}^*} + 1$. The rest of the claim follows by the induction hypothesis.

    Given the claim, we can bound $\OPT_{\mathrm{proxy}}$ as follows: for any $t \in [n]$, suppose $\overline{s} = {\arg\min}_{s \le t}c^*_s$, so that $d^*_{t} = c^*_{\overline{s}}$, and by the claim,
    \[
        \frac{q^*_t}{d_t^*} = \frac{q^*_t}{c_{\overline{s}}^*} \le q_t^* \left(\frac{1}{c^*_t} + (t - \overline{s})\right) \le \frac{q_t^*}{c_t^*} + tq_t^*.
    \]
    Therefore,
    \[
        \OPT_{\mathrm{proxy}} = \sum_{t \in [n]} \frac{q_t^*}{d_t^*} \le \sum_{t \in [n]} \frac{q_t^*}{c_t^*} + \sum_{t \in [n]} tq_t^* = \OPT + \OPT.
    \]
\end{proof}

With this lemma in hand, let us go back to the bound (\ref{eqn: updated-greedy-upper-bound-2}): we have
\begin{align*}
    \ugreedy &\le \sum_{i = 1}^K \frac{1}{c^*_{t_{i}}} (r_{j_{i - 1} + 1} - r_{j_i + 1}) \le \sum_{i = 1}^K \frac{1}{d^*_{t_{i}}} (r_{j_{i - 1} + 1} - r_{j_i + 1}) & (d^*_{t_i} \le c^*_{t_i}), \\
    &= \sum_{i = 1}^K r_{j_{i - 1} + 1} \left(\frac{1}{d^*_{t_i}} - \frac{1}{d^*_{t_{i - 1}}} \right) & (\text{rearranging the sum}), \\
    &\le \sum_{i = 1}^K r^*_{t_{i}} \left(\frac{1}{d^*_{t_i}} - \frac{1}{d^*_{t_{i - 1}}} \right) & (\text{Lemma } \ref{lem: updated-greedy-measuring-progress}), \\
    &= \sum_{i = 1}^K \frac{1}{d^*_{t_{i}}} (r^*_{t_{i}} - r^*_{t_{i + 1}}) & (\text{rearranging the sum}), \\
    &= \sum_{i = 1}^K \sum_{t = t_{i}}^{t_{i + 1} - 1} \frac{q^*_{t}}{d^*_{t_i}} \\
    &\le \sum_{i = 1}^K \sum_{t = t_{i}}^{t_{i + 1} - 1} \frac{q^*_{t}}{d^*_{t}} & (t \ge t_i), \\
    &
    \le \OPT_{\mathrm{proxy}} \le 2 \ \OPT & (\text{Lemma} \ \ref{lem: monotone-conditional-probability}).
\end{align*}

This proves \cref{thm: cna-2-approximation}.

\section{Omitted details from \cref{sec: graph-probing}}\label{app: graph-probing-omitted-proofs}

We provide algorithms and proofs omitted from \cref{sec: graph-probing} here. 

\subsection{Adaptivity gap}

First, we prove \cref{lem: adaptivity-gap-graph-probing}:

\AdaptivityGapGraphProbing*

\begin{proof}
    Our graph $G = (V, E)$ is a union of $\frac{n}{k + 1}$ stars where we set $k = \sqrt{n}$. For each of the $\frac{n}{k + 1}$ stars, the central vertex $v$ is associated with the Bernoulli random variable $Y_{v}$ with $\Pr(Y_{v} = 1) = p := \frac{\log n}{\sqrt{n}}$. For each vertex $v$ that is a leaf vertex of some star, $\Pr(Y_v = 1) = 0$. That is, for each edge $e \in E$ in the graph, $\Pr(X_e = 1) = p$. Further, all edges in a given star are simultaneously either $1$ or $0$.

    Consider the following adaptive algorithm: go through the stars one-by-one and probe any edge in the star, until a non-zero edge is found. If a non-zero edge is found in a star, probe all other edges of that star. If no nonzero edge is found in any star, probe every other edge in the graph. We claim that the cost of this algorithm is $O\left(\sqrt{\frac{n}{\log n}}\right)$. If there is a nonzero edge in some star, it takes at most $\frac{n}{k + 1}$ iterations to find it. The algorithm therefore terminates in $\frac{n}{k + 1} + (k - 1) = O(\sqrt{n})$ probes. Further, the probability that there is no star that has a nonzero edge is $(1 - p)^{\frac{n}{k + 1}} \le \exp\left(\frac{pn}{k + 1}\right) \le \exp\left(-\log n\right) = \frac{1}{n}$. Since there are $n - \frac{n}{k + 1} = \frac{nk}{k + 1}$ edges in the graph, this algorithm has expected cost $\le (1 - 1/n) \cdot O(\sqrt{n}) + (1/n) \cdot \frac{nk}{k + 1} = O(\sqrt{n})$.

    Now we show that any nonadaptive algorithm must use $\Omega\left(\frac{n}{\log n}\right)$ probes on average. A nonadaptive algorithm must fix a (possibly randomized) sequence of probes in advance. Denote $T = \frac{k}{2p} = \frac{n}{2 \log n}$. For any fixed (deterministic) sequence of edge probes, the expected number of nonzero edges in the first $T$ edges of the sequence is $Tp = \frac{k}{2}$. Therefore, by Markov's inequality, the probability that there are $\ge k$ nonzero edges among the first $T$ edges is $\le \frac{(k/2)}{k} = \frac{1}{2}$. That is, any fixed sequence of probes has cost $\ge T$ with probability $\frac{1}{2}$, i.e., must have cost $\ge \frac{T}{2} = \Omega\left(\frac{n}{\log n}\right)$ on average. Finally, the average cost of a (randomized) sequence is simply the a weighted average of the costs of some deterministic sequence, and therefore must be $\Omega\left(\frac{n}{\log n}\right)$ as well.

    Therefore, the adaptivity gap is $\Omega\left(\frac{1}{\sqrt{n}} \times \frac{n}{\log n}\right) = \Omega\left(\frac{\sqrt{n}}{\log n}\right)$.
\end{proof}

\subsection{From Graph Probing to Vertex Probing}\label{sec: graph-probing-to-vertex-probing}

Here, we convert adaptive \gp algorithms to adaptive Vertex Probing algorithms.

\GraphProbingToVertexProbing*

Our transformation is as follows: each time the \gp algorithm decides to probe edge $uv$, the corresponding Vertex Probing algorithm probes both vertex $u$ and vertex $v$ (in any order), unless they have already been probed. Then, the Vertex Probing algorithm feeds back the outcome $X_{uv} = Y_u \lor Y_v$ for the edge $uv$ back to the \gp algorithm, so that it can adapt. This is formalized in Algorithm \ref{alg: graph-probing-to-vertex-probing}.

\begin{algorithm}[t]
    \caption{\textsc{GraphProbingToVertexProbing$(\ALG^e, G, k)$}}\label{alg: graph-probing-to-vertex-probing}
    \begin{algorithmic}[1]
        \Statex \textbf{input}: algorithm $\ALG^e$ for \gp, graph $G$, and knapsack size $k$
        \State initialize $\ALG^e$ with input graph $G$ and knapsack size $k$
        \State $S \gets \emptyset$ \Comment{Probed vertices}
        \While{$\sum_{v \in S} \deg_v Y_v < k$}
            \State let $uv$ be the edge that $\ALG^e$ probes next
            \State if $u \not\in S$, probe $u$, and add $u$ to $S$ \Comment{Determine $Y_u$}
            \State if $v \not\in S$, probe $v$, and add $v$ to $S$ \Comment{Determine $Y_v$}
            \State report the outcome $X_{uv} = Y_u \lor Y_v$ of edge $uv$ to $\ALG^e$
            \State if $\ALG^e$ terminates, \textbf{return} $S$
        \EndWhile
    \State \textbf{return} $S$
    \end{algorithmic}
\end{algorithm}

\begin{proof}[Proof of Lemma \ref{lem: transformation-graph-probing-to-vertex-probing}]
    We need to prove the following:
    \begin{enumerate}
        \item (Correctness) If $\sum_{v \in V} \deg_v Y_v \ge k$, then the Vertex Probing algorithm (denoted $\ALG^v$) probes vertex set $S$ such that $\sum_{v \in S} \deg_v Y_v \ge k$.
        \item (Cost) If the \gp algorithm $\ALG^e$ probes $N$ edges, then the Vertex Probing algorithm probes $\le 2N$ vertices.
    \end{enumerate}

    (Correctness) $\ALG^v$ only terminates if either $\sum_{v \in S} \deg_v Y_v \ge k$, or if $\ALG^e$ terminates. The \gp algorithm $\ALG^e$ only terminates if it has probed some edge set $F \subseteq E$ such that $\sum_{e \in F} X_e \ge k$, or if it has probed all edges.
    
    Suppose $\sum_{e \in F} X_e \ge k$ for the edge set $F$ that $\ALG^e$ probes, and denote $F_1 = \{e \in F: X_e = 1\}$ to be the set of non-zero edges. $\ALG^v$ probes each endpoint of each edge $e \in F$, that is, $S = \{v \in V: v \text{ is an endpoint of some } e \in F\}$. Denote $S_1 = \{v \in S: Y_v = 1\}$ to be the set of non-zero vertices. Note that for each edge $uv \in F_1$, since $1 = X_{uv} = Y_{u} \lor Y_v$, we must have that $Y_u = 1$ or $Y_v = 1$, or that $u \in S_1$ or $v \in S_1$. Therefore, $uv$ is counted at least once in the sum $\sum_{v \in S} \deg_v Y_v = \sum_{v \in S_1} \deg_v$. Consequently, $\sum_{v \in S} \deg_v Y_v \ge |F_1| = \sum_{e \in F} X_e$.
    
    Now suppose that $\ALG^e$ has probed all edges and terminates. Then, $\ALG^v$ probes all vertices, and can therefore terminate.
    
    (Cost) Clearly, if $\ALG^e$ probes $N$ edges, then $\ALG^v$ probes $\le 2N$ vertices since it only probes the endpoints of edge edge that $\ALG^e$ probes. This proves Lemma \ref{lem: transformation-graph-probing-to-vertex-probing}.
\end{proof}

\subsection{From Vertex Probing to Graph Probing}\label{sec: vertex-probing-to-graph-probing}

Next, we show that any Vertex Probing algorithm can be converted into a \gp algorithm with bounded cost (see Algorithm \ref{alg: vertex-probing-to-graph-probing} for a formal description). 
Consequently, using the $3$-approximation for Vertex Probing from \cite{deshpande2016approximation} yields Theorem \ref{thm: graph-probing}.

\begin{algorithm}[t]
    \caption{\textsc{VertexProbingToGraphProbing}$(\ALG^v, G, k)$}\label{alg: vertex-probing-to-graph-probing}
    \begin{algorithmic}[1]
        \Statex \textbf{input}: algorithm $\ALG^v$ for Vertex Probing, graph $G$, and knapsack size $k$
        \State initialize $F \gets \emptyset$ \Comment{Probed edges}
        \State initialize $k' \gets k$ \Comment{Number of non-zero edges to be found}
        \State initialize $V' \gets V$, $E' \gets E$, $G' = (V', E')$ \Comment{$G'$ is modified by the algorithm}
        \State $b = (\ast, \ldots, \ast) \in \R^V$ \Comment{Outcomes for each vertex $v \in V$; $*$ indicates unknown}
        \While{not all edges have been probed}
            \State initialize $\ALG^v$ with input graph $G'$ and knapsack size $k'$
            \While{$\ALG^v$ does not terminate}
                \State let $v \in V$ be the vertex that $\ALG^v$ probes next
                \State if $v \in V'$, run \textsc{ProbeVertex}$(G', v, k', b)$ \Comment{probes some edges and updates $G', k', b$}
                 \State if $k' = 0$, \textbf{return} \Comment{$k$ non-zero edges found}
                \State report the outcome $b_v$ for vertex $v$ to $\ALG^v$ 
            \EndWhile
        \EndWhile
        \State probe any edges $e \in E$ that have not been probed
        \State \textbf{return}
    \end{algorithmic}
\end{algorithm}

\subsubsection{Algorithm}

In Algorithm \ref{alg: vertex-probing-to-graph-probing}, we initialize the graph $G' = (V', E')$ as the given graph $G = (V, E)$, and $k'$ (number of nonzero edges remaining to be found) as $k$. As the algorithm probes edges, $k'$ is updated and appropriate edges and vertices are removed from $G'$. Graph $G'$ is allowed to have multiple self-loops at a vertex; this is to simplify our algorithm and analysis.

The algorithm is divided into several \emph{phases}. In each phase, a Vertex Probing algorithm $\ALG^v$ is initialized on $G'$ and $k'$, and we show that either (1) $k'$ is reduced by factor $2$ by the end of each phase, or (2) the algorithm terminates (Lemma \ref{lem: graph-probing-half-bag-size}). Further, we maintain the invariant that vertex random variables $\{Y_v, v \in V'\}$ in the \emph{residual graph} $G' = (V', E')$ are independent of each other and of the outcomes of any edges probed so far (Lemma \ref{lem: graph-probing-residual-instance-is-independent-main}). Further, $\deg_v(G) = \deg_{v}(G')$ for all remaining vertices $v \in V'$.

In a given phase, $\ALG^v$ directs us to probe some vertex $v \in V'$. We cannot probe a vertex directly, so we simply probe the edges incident on $v$ one by one, until a $0$ edge is found or until all edges have been probed and were all non-zero (Algorithm \ref{alg: probe-vertex}, \textsc{ProbeVertex}). In either case, if we stop here and do not probe more edges, then the remaining vertices are no longer independent. Therefore, we recursively call \textsc{ProbeVertex} on those neighbors of $v$ for which the corresponding edge was $1$. 

Once \textsc{ProbeVertex} on $v$ is finished, we need to report back the outcome $b_v$ of random variable $Y_v$ so that $\ALG^v$ can adapt. If some edge $uv$ incident on $v$ is $0$, then clearly $Y_v \le Y_u \lor Y_v = X_{uv} = 0$, i.e., we can report back $b_v = 0$. However, if all edges incident on $v$ are $1$, then we cannot be sure of the outcome for $Y_v$, since it is possible that $Y_v = 1$, or that $Y_v = 0$ and $Y_u = 1$ for each neighbor $u$ of $v$. In this case, we report $b_v = 1$. This could potentially mislead the Vertex Probing algorithm $\ALG^v$, but we show in Lemma \ref{lem: probe-vertex-sees-more-ones-than-zeroes} that this can only reduce the expected number of probes by $\ALG^v$.

The final component of our main algorithm is the subroutine \textsc{RemoveVertex}, which rearranges the residual graph $G'$ when $X_{uv} = 0$ for some edge $uv \in E'$. Essentially, \textsc{RemoveVertex}$(v)$ removes the vertex $v$ from $V'$ and converts any other (unprobed) edge incident on $v$ into a loop at its other endpoint.

\begin{algorithm}[t]
    \caption{\textsc{ProbeVertex}$(G', v, k', b)$}\label{alg: probe-vertex}
    \begin{algorithmic}[1]
        \Statex \textbf{input}: graph $G' = (V', E')$ (possibly with loops), vertex $v \in V'$, and integer $k' > 0$
        \For{each loop $l$ at $v$}
            \State probe random variable $X_{l}$ and remove $l$ from $E'$
            \If{$X_l = 0$}
                \State \textsc{RemoveVertex}$(G', v)$
                \State update $b_v \gets 0$ and \textbf{return}
            \Else
                \State $k' \gets k' - 1$; \textbf{return} if $k' = 0$
            \EndIf
        \EndFor
        \State let $u_1, \ldots, u_N$ be the neighbors of $v$ in $V'$
        \For{$i = 1$ to $N$}
            \State probe edge $vu_i$ and remove $vu_i$ from $E'$
            \If{$X_{vu_i} = 0$}
                \State \textsc{RemoveVertex}$(G', v)$ and \textsc{RemoveVertex}$(G', u_i)$
                \For{$j < i$}
                    \State \textsc{ProbeVertex}$(G', u_j, k', b)$
                \EndFor
                \State update $b_v \gets 0$, $b_{u_i} \gets 0$ and \textbf{return}
            \Else
                \State $k' \gets k' - 1$; \textbf{return} if $k' = 0$
            \EndIf
        \EndFor
        \State remove $v$ from $V'$ and update $b_v \gets 1$
        \For{$i = 1$ to $N$}
            \State \textsc{ProbeVertex}$(G', u_i, k', b)$
        \EndFor
        \State \textbf{return}
    \end{algorithmic}
    \noindent\rule{\linewidth}{0.4pt}
    \textsc{RemoveVertex}$(G', v)$
    \begin{algorithmic}[1]
        \Statex \textbf{input}: graph $G' = (V', E')$ (possibly with loops), and vertex $v \in V'$
        \State remove each loop $l$ at $v$ from $E'$
        \For{each neighbor $u$ of $v$ in $G'$}
            \State replace edge $uv$ in $E'$ with a loop at $u$
        \EndFor
        \State remove $v$ from $V'$
        \State \textbf{return}
    \end{algorithmic}
\end{algorithm}

\subsubsection{Analysis}

We present the analysis of the algorithm now. All omitted proofs are deferred to Appendix \ref{app: graph-probing-omitted-proofs}. Throughout, we denote $G = (V, E)$ as the original graph and $G' = (V', E')$ as the residual graph. $G$ stays fixed while $G'$ is updated during the algorithm. Similarly, $k$ is the initial knapsack size (fixed), while $k'$ is the (remaining) knapsack size at any instant in the algorithm. Further, recall that we divide the algorithm into different phases (a phase corresponding to one iteration of lines 6 to 11 in the outer while loop in Algorithm \ref{alg: vertex-probing-to-graph-probing}), and initialize $\ALG^v$ once in each phase.

Our first lemma establishes the independence of the vertex random variables $\{Y_v: v \in V'\}$ in the residual graph from the outcomes of probed edges after each run of \textsc{ProbeVertex} (and in particular after each phase of the algorithm).

\ResidualInstance*

To prove this lemma, we first prove in the following lemma that random variables $\{Y_v: v \in V'\}$ are independent of the outcome of probed edges $F \subseteq E$ if no edge $e \in F$ is incident on any vertex $v \in V'$.

\begin{proof}
    Let $F \subseteq E$ be the set of edges probed with outcomes $x_e \in \{0, 1\}$ for each edge $e \in F$. We need to show that for all outcomes $b' \in \{0, 1\}^{V'}$ for the vertices $V'$,
    \[
        \Pr(Y_v = b'_v \ \forall \ v \in V' \ \text{and} \ X_e = x_e \ \forall \ e \in F) = \Pr(X_e = x_e \ \forall \ e \in F) \times \prod_{v \in V'} \Pr(Y_v = b'_v).
    \]
    The outcome for any edge only depends on its endpoints, and therefore knowing the outcome of edges not incident on vertices in $V'$ does not affect the distribution of those vertices.

    Formally, let $B = \{0, 1\}^{V \setminus V'}$ denote the set of all outcomes for vertices other than $V'$. Then fixing outcome $b \in B$ for the vertices $V \setminus V'$ completely determines the outcomes for edges $F$. Let $\overline{B} \subseteq B$ denote the outcomes for vertices compatible with the outcomes for $F$, i.e., for which $X_e = x_e$ for each $e \in F$. Then, we have
    \begin{align*}
        &\Pr(Y_v = b'_v \, \forall \, v \in V', X_e = x_e \, \forall \, e \in F) 
        \,\, = \,\, \sum_{b \in \overline{B}} \Pr(Y_v = b'_v \, \forall \, v \in V', Y_u = b_u \, \forall \, u \in V \setminus V') \\
        & \qquad = \sum_{b \in \overline{B}} \left(\prod_{v \in V'}\Pr(Y_v = b'_v)\right)  \Pr(Y_u = b_u \ \forall \ u \in V \setminus V') \\
         & \qquad = \,\, \left(\prod_{v \in V'}\Pr(Y_v = b'_v) \right) \times \Pr(X_e = x_e \ \forall \ e \in F),
    \end{align*}
    as desired. Above, the second equality follows since the $Y_v$ variables are independent.
\end{proof}

The following lemma states that for each probed edge $e \in F$, neither of its endpoints is in the vertex set $V'$ of the residual graph. This proves Lemma \ref{lem: graph-probing-residual-instance-is-independent-main}.

\begin{lemma}\label{lem: residual-graph-is-not-adjacent-to-probed-edges}
    Consider graph $G' = (V', E')$ at the end of each run of \textsc{ProbeVertex} in \textsc{VertexProbingToGraphProbing}. Then, for any vertex $v \in V'$, none of the edges incident on $v$ in the original graph $G = (V, E)$ have been probed.
\end{lemma}

\begin{proof}
    Suppose, to the contrary, that some edge $uv \in E$ has been probed.
    $uv$ could have been probed only if \textsc{ProbeVertex} is run on one of its endpoints. If \textsc{ProbeVertex}$(v)$ is run, then $v$ is removed from $V'$, which is a contradiction. If $uv$ was probed during \textsc{ProbeVertex}$(u)$ and $X_{uv} = 1$, then \textsc{ProbeVertex}$(v)$ must also have been run, again a contradiction. Therefore, we must have $X_{uv} = 0$. But in this case, both $u$ and $v$ are removed from $V'$, which again yields a contradiction. Therefore, no edge $uv \in E$ is probed if $v \in V'$.
\end{proof}

\begin{lemma}\label{lem: probe-vertex-sees-more-ones-than-zeroes}
    Let $m_1$ and $m_0$ denote the number of non-zero and zero edges probed during a run of \textsc{ProbeVertex}. Then $m_0 \le m_1 + 1$.
\end{lemma}

\begin{proof}
    Suppose \textsc{ProbeVertex} is run on graph $G' = (V', E')$ and vertex $v \in V'$. Each run of \textsc{ProbeVertex}$(v)$ can recursively call \textsc{ProbeVertex} on other vertices adjacent to $v$. Consider the (rooted, directed) recursion tree $\tau = (V'', A)$ defined as follows: $V'' \subseteq V'$ is the subset of vertices on which \textsc{ProbeVertex} is called, and there is a directed edge from $u \in V''$ to $w \in V''$ in $A$ if \textsc{ProbeVertex}$(w)$ is called during \textsc{ProbeVertex}$(u)$.

    If \textsc{ProbeVertex}$(u)$ probes edge $uw \in E$ and $X_{uw} = 1$, then $u$ is the parent of $w$ in $\tau$ (i.e., $(u, w) \in A$). That is, $m_1 = |A|$. Further, as soon as some edge incident on a vertex $u \in V''$ is probed and turns out to be $0$, no other edges incident on $u$ are probed. Therefore, $m_0 \le |V''|$. Since $\tau$ is a tree, $|A| = |V''| - 1$. That is, $m_1 \ge m_0 - 1$.
\end{proof}

The next lemma shows that the algorithm makes at most $O(\OPT^e(G, k))$ calls (in expectation) to \textsc{ProbeVertex} in the first of the various phases of the algorithm:
\begin{lemma}\label{lem: graph-probing-vertex-actual-vs-reported-outcome}
    Consider algorithm $\ALG^v$ that is an $\alpha$-approximation for Vertex Probing, for some $\alpha \ge 1$. Then the number of calls to \textsc{ProbeVertex} made in the first phase of \textsc{VertexProbingToGraphProbing} is at most $2\alpha \ \OPT^e(G, k)$.
\end{lemma}

First, we need the following lemmas:

\begin{lemma}
    In Algorithm \textsc{VertexProbingToGraphProbing}, for any vertex $v \in V$, if $b_v \neq \ast$ (i.e., if the outcome $b_v$ is not unknown), then $b_v \ge Y_v$. Further, $b_v = Y_v + 1$ if and only if $X_e = 1$ for each edge $e \in E$ incident on $v$.
\end{lemma}

\begin{proof}
    If $b_v \neq \ast$, i.e., if $b_v$ is updated by \textsc{ProbeVertex} at some point, then $b_v = 0$ or $b_v = 1$. \textsc{ProbeVertex} updates $b_v = 0$ only if $X_{uv} = 0$ for some edge $uv \in E$. Since $X_{uv} = Y_u \lor Y_v \ge Y_v$, this implies that $Y_v = 0$. Further, \textsc{ProbeVertex} updates $b_v = 1$ only if $X_e = 1$ for each edge $e \in E$ incident on $v$.
\end{proof}

\begin{lemma}\label{lem: graph-probing-misreporting-only-helps}
    Consider an algorithm $\ALG^v$ for the Vertex Probing problem. Then we can assume without loss of generality that the cost of $\ALG^v$ on vertex set $V$ and knapsack size $k$ is at least the cost of $\ALG^v$ on vertex set $V \setminus \{v\}$ and knapsack size $k - \deg_v$.
\end{lemma}

\begin{proof}
    First, note that we can estimate the expected cost of $\ALG^v$ on any input vertex set $\hat{V}$ with degrees $\deg_u, u \in \hat{V}$ and knapsack size $\hat{k} > 0$. To do so, we can simply sample independent outcomes $Y_u: u \in \hat{V}$ with given probabilities $\Pr(Y_u = 1)$ and run $\ALG^v$ for this scenario. Do this $N$ times, with $Z_i$ denoting the cost of the algorithm for the $i$th run for $i \in [N]$. Denote $Z = \frac{1}{N} \sum_{i \in [N]} Z_i$. Then the expected cost of the algorithm is precisely $\E Z$, and standard concentration bounds show that $Z$ is arbitrarily close to $\E Z$ in a polynomial number of runs $N$.

    Given this, we prove the lemma statement. Suppose $\ALG^v(V, k) < \ALG^v(V \setminus \{v\}, k - \deg_v)$ (both these numbers can be estimated as described above). Then we can simply replace $\ALG^v(V \setminus \{v\}, k - \deg_v)$ by $\ALG^v(V, k)$, except that when $\ALG^v(V, k)$ probes vertex $v$, we report $Y_v = 1$ with probability $\Pr(Y_v = 1)$ and $Y_v = 0$ with probability $(Y_v = 0)$. The cost of this algorithm is clearly the expected cost of $\ALG^v(G, k)$. Further, $\ALG^v(V, k)$ either probes each vertex $v \in V$ or fills the knapsack of size $k$. Vertices other than $v$ still fill a knapsack of size $k - \deg_v$.
\end{proof}

We are ready to prove Lemma \ref{lem: graph-probing-vertex-actual-vs-reported-outcome}.

\begin{proof}[Proof of Lemma \ref{lem: graph-probing-vertex-actual-vs-reported-outcome}]
    In the first phase of the algorithm, $\ALG^v$ is initialized with graph $G' = G$ and knapsack size $k' = k$. In each iteration, $\ALG^v$ seeks to probe some vertex $v \in V$ and expects the realized value of $Y_v$ in return. Then, it updates the bag size as $k' \gets k' - Y_v \deg_v$.

    In each iteration, \textsc{VertexProbingToGraphProbing} makes at most $1$ call to \textsc{ProbeVertex} and informs $\ALG^v$ of the outcome $b_v$ for the random variable $Y_v$ by probing edges incident on $v$. Therefore, the number of calls to \textsc{ProbeVertex} is the first phase is at most the number of vertices that $\ALG^v$ seeks to probe.

    By the previous lemma, $b_v \neq Y_v$ if and only if $b_v = 1$, $Y_v = 0$, and $X_e = 1$ for each edge $e \in E$ incident on $v$. $\ALG^v$ then incorrectly updates $k' \gets k' - \deg_v$. However, by Lemma \ref{lem: graph-probing-misreporting-only-helps}, this does not increase the number of probes that $\ALG^v$ makes. That is, misreporting $Y_v$ does not increase the number of calls made by $\ALG^v$.

    Finally, $\ALG^v(G, k) \le \alpha \ \OPT^v(G, k) \le 2 \alpha \ \OPT^e(G, k)$, where the first inequality holds because $\ALG^v$ is an $\alpha$-approximation for Vertex Probing and the second inequality holds by Lemma \ref{lem: transformation-graph-probing-to-vertex-probing}.
\end{proof}

Hereafter, we assume that $\ALG^v$ is the algorithm from \cite{deshpande2016approximation} that is a $3$-approximation for Stochastic min-Knapsack (and therefore $3$-approximation for Vertex Probing).

Therefore, we get the following stronger version of \cref{lem: graph-probing-first-phase-simple}, which bound the cost of the algorithm in the first phase.

\begin{lemma}\label{lem: graph-probing-first-phase}
    Let $M_0, M_1$ be the number of zeros, ones seen in the first phase of \textsc{VertexProbingToGraphProbing} on input graph $G$ and knapsack size $k$. Then $M_0 \le M_1 + 6 \ \OPT^e(G, k)$. Consequently, the number of probes in the first phase of the algorithm is at most $8 \ \OPT^e(G, k)$.
\end{lemma}

\begin{proof}
    Each phase of the algorithm consists of several calls to \textsc{ProbeVertex} made by $\ALG^v(G', k')$. By Lemma \ref{lem: graph-probing-vertex-actual-vs-reported-outcome}, the number of such calls is at most $6 \ \OPT^e(G, k)$.

    Let $M_1$ (resp. $M_0$) be the number of edges probed in phase $1$ that are $1$ (resp. $0$). Then by the above and Lemma~\ref{lem: probe-vertex-sees-more-ones-than-zeroes}, we get $M_0 \le M_1 + 6 \ \OPT^e(G, k)$. 

    The total number of probes in phase $1$ is therefore $M_0 + M_1 \le 2 M_1 + 6 \ \OPT^e(G, k)$. However, $M_1 \le k \le \OPT^e(G, k)$ since $\OPT^e$ must probe at least $k$ edges (in case there are $k$ non-zero edges) or exactly $|E| \ge k$ edges.
\end{proof}

Next, we prove that each phase of the algorithm reduces the number of nonzero edges to be found by at least a factor $2$. This completes the proof of \cref{thm: graph-probing}.

\HalfBagSize*

\begin{proof}
    Let $S \subseteq V$ be the set of vertices probed by $\ALG^v(G, k)$ in the first phase. $\ALG^v$ terminates if either $S = V$ or if $\sum_{v \in S} b_v \deg_v \ge k$. If $S = V$, then every edge in $G$ has been probed, and so \textsc{VertexProbingToGraphProbing} also terminates.

    So suppose $\sum_{v \in S} b_v \deg_v \ge k$. Further, by construction in \textsc{ProbeVertex}, if $b_v = 1$ for some $v \in S$, then every edge $e$ incident on $v$ has been probed and $X_{e} = 1$. Therefore, for the set $F$ of edges probed in phase 1,
    \begin{align*}
        2 \sum_{e \in F} X_e &\ge  \sum_{v \in S: b_v = 1} \sum_{e: e \ni v} X_e \ge \sum_{v \in S: b_v = 1} \deg_v = \sum_{v \in S} b_v \deg_v \ge k.
    \end{align*}
    Therefore, at least $\frac{k}{2}$ non-zero edges have been found in phase 1.
\end{proof}

\end{document}